\documentclass[a4paper,UKenglish,cleveref, autoref, thm-restate]{lipics-v2021}

\usepackage[composition,notion, quotation,xcolor]{knowledge}

\usepackage{mathtools}
\usepackage{xspace}

\usepackage{tcolorbox}

\usepackage{booktabs}%better tables
\usepackage{tikz,xcolor,hyperref}
\usetikzlibrary{positioning, arrows}
\usetikzlibrary{patterns}
\usetikzlibrary{shapes}
\usetikzlibrary{decorations.pathreplacing}
\usetikzlibrary{fadings}
\usetikzlibrary{intersections}

\usepackage{todonotes}
\setuptodonotes{inline}

\usepackage{lineno}
\knowledgenewrobustcmd{\Value}[1]{\cmdkl{\mathsf{Val}_{#1}}}

\knowledgenewrobustcmd{\play}{\cmdkl{\mathsf{play}}}
\knowledgenewrobustcmd{\wt}{\cmdkl{wt}}
\knowledgenewrobustcmd{\Nind}[1]{\cmdkl{N({#1})}}

\newcommand{\Xflawed}{\Xcal_{\text{flaw}}}
\knowledgenewrobustcmd{\slopemax}{\cmdkl{s_{\max{}}}}
\knowledgenewrobustcmd{\wmax}{\cmdkl{w_{\max{}}}}
\knowledgenewrobustcmd{\MM}{\cmdkl{M}}

\newcommand{\Tmax}{\Tcal}

\knowledgenewrobustcmd{\Reg}[1]{\cmdkl{\mathsf{Reg}_{#1}}}
\knowledgenewrobustcmd{\Reglt}[1]{\cmdkl{\mathsf{Reg}_{#1}^<}}

\knowledgenewrobustcmd{\DDelta}{\cmdkl{\Delta}}
\knowledgenewrobustcmd{\kkappa}{\cmdkl{\kappa}}

\newcommand{\Maxsf}{\mathsf{Max}}
\newcommand{\Minsf}{\mathsf{Min}}
\newcommand{\Mamax}{\Maxsf}
\newcommand{\Mimin}{\Minsf}
\newcommand{\Player}{\mathsf{P}}
\knowledgenewrobustcmd{\upclock}[1]{\cmdkl{X_{#1}^\uparrow}}
\newcommand{\weight}{\mathsf{weight}}
\knowledgenewrobustcmd{\opt}[2]{\cmdkl{\mathsf{opt}^{#1}_{#2}}}
\knowledgenewrobustcmd{\Opt}[2]{\cmdkl{\mathsf{Opt}^{#1}_{#2}}}
\knowledgenewrobustcmd{\Next}{\cmdkl{\mathsf{Leaving}}}
\knowledgenewrobustcmd{\reg}[1]{\cmdkl{\mathsf{reg}(#1)}}
\knowledgenewrobustcmd{\adh}[1]{\cmdkl{\overline{#1}}}

\newcommand{\wtg}{\textsc{wtg}\xspace}
\newcommand{\Wtg}{\textsc{Wtg}\xspace}
\newcommand{\wtgs}{\textsc{wtg}s\xspace}

\knowledgenewrobustcmd{\RC}[1]{\cmdkl{\Rcal(#1)}}
\knowledgenewrobustcmd{\TC}[1]{\cmdkl{\Tcal(#1)}}
\knowledgenewrobustcmd{\Gtilde}{\cmdkl{\tilde{\Gcal}}}

\newcommand{\ie}{i.e.,~}
\newcommand{\etal}{\textit{et al.}~}

\usetikzlibrary{decorations,decorations.markings} 
\pgfkeys{/tikz/.cd,
    contour distance/.store in=\ContourDistance,
    contour distance=-10pt, % for the other orientation use a +
    contour step/.store in=\ContourStep,
    contour step=1pt,
}

\pgfdeclaredecoration{closed contour}{initial}
{% 
\state{initial}[width=\ContourStep,next state=cont] {
    \pgfmoveto{\pgfpoint{\ContourStep}{\ContourDistance}}
    \pgfcoordinate{first}{\pgfpoint{\ContourStep}{\ContourDistance}}
    \pgfpathlineto{\pgfpoint{0.3\pgflinewidth}{\ContourDistance}}
    \pgfcoordinate{lastup}{\pgfpoint{1pt}{\ContourDistance}}
    
  }
  \state{cont}[width=\ContourStep]{
     \pgfmoveto{\pgfpointanchor{lastup}{center}}
     \pgfpathlineto{\pgfpoint{\ContourStep}{\ContourDistance}}
     \pgfcoordinate{lastup}{\pgfpoint{\ContourStep}{\ContourDistance}}
  }
  \state{final}[width=\ContourStep]
  { % perhaps unnecessary but doesn't hurt either
    \pgfmoveto{\pgfpointanchor{lastup}{center}}
    \pgfpathlineto{\pgfpointanchor{first}{center}}
  }
}

\newcommand{\eqdef}{\stackrel{\mathsf{def}}{=}}

\newcommand{\set}[1]{\left\{\; {#1} \;\right\}}

\newcommand{\setof}[2]{\left\{{#1} \;:\; {#2} \right\}}

\newcommand{\Nat}{\mathbb{N}}

\newcommand{\DefineNewOrder}[4]{%
	\expandafter\newcommand\csname#1le\endcsname{\mathrel{#2}_{#4}}%
	\expandafter\newcommand\csname#1ge\endcsname{\mathrel{\reflectbox{\text{$#2$}}}_{#4}}%
	\expandafter\newcommand\csname#1leq\endcsname{\mathrel{#3}_{#4}}%
	\expandafter\newcommand\csname#1geq\endcsname{\mathrel{\text{\reflectbox{$#3$}}}_{#4}}%
	\expandafter\newcommand\csname#1eq\endcsname{\mathrel{\equiv}_{#4}}%
	\expandafter\newcommand\csname#1perp\endcsname{\mathrel{\perp}_{#4}}%
}

\newcommand{\Nbb}{\ensuremath{\mathbb{N}}}

\newcommand{\Rbb}{\ensuremath{\mathbb{R}}}

\newcommand{\Gcal}{\ensuremath{\mathcal{G}}}

\newcommand{\Mcal}{\ensuremath{\mathcal{M}}}

\newcommand{\Pcal}{\ensuremath{\mathcal{P}}}

\newcommand{\Rcal}{\ensuremath{\mathcal{R}}}
 
\newcommand{\Tcal}{\ensuremath{\mathcal{T}}}

\newcommand{\Xcal}{\ensuremath{\mathcal{X}}}

\newcommand{\ignore}[1]{}%ignore le charactère suivant.
\newcommand{\tiff}{\textit{iff} }

\if@anglais

\else

\fi

\newcommand{\intxy}[4]{\ensuremath{\mathchoice
    {\left#3\protect\TestGauche{#1},  \protect\TestDroite{#2}\right#4}
    {\left#3\protect\TestGauche{#1},  \protect\TestDroite{#2}\right#4}
    {\left#3\protect\TestGauche{#1},\protect\TestDroite{#2}\right#4}
    {\left#3\protect\TestGauche{#1},\protect\TestDroite{#2}\right#4}
}}  
	\newcommand{\intff}[2]{\intxy{#1}{#2}{[}{]}}
	\newcommand{\intof}[2]{\intxy{#1}{#2}{(}{]}}
	\newcommand{\intfo}[2]{\intxy{#1}{#2}{[}{)}}

\newcommand{\TestGauche}[1]{\ifthenelse{\equal{#1}{}}{\minf}{#1}}
\newcommand{\TestDroite}[1]{\ifthenelse{\equal{#1}{}}{\pinf}{#1}}

\newcommand{\enstq}[2]{\left\{ \left.#1\ \vphantom{#2}\right|\ #2  \right\}}
\newcommand{\pa}[2][9]{%
	\ifthenelse{#1 = 0}{#2}{}%
	\ifthenelse{#1 = 1}{(#2)}{}%
	\ifthenelse{#1 = 2}{\big(#2\big)}{}%
	\ifthenelse{#1 = 3}{\Big(#2\Big)}{}%
	\ifthenelse{#1 = 4}{\bigg(#2\bigg)}{}%
	\ifthenelse{#1 = 5}{\Bigg(#2\Bigg)}{}%
	\ifthenelse{#1 = 9}{\left(#2\right)}{}%
}

\newcommand{\distinf}[1]{\left\|#1\right\|_{\infty}}

\knowledge{notion}
 | Clock constraints
 | clock constraints
 | guards
 | guard
\knowledge{notion}
 | almost non-Zeno \wtg 
 | almost non-Zeno \wtgs 
 | almost non-Zeno
\knowledge{notion}
 | output edges
 | output edge
\knowledge{notion}
 |Value problem
\knowledge{notion}
 |Approximation problem
 |Approximation
 |Approximable
 |Approximability

\knowledge{notion}
 | kernel
 | kernels
\knowledge{notion}
 | green
 | green location
 | \emph {green} transition
 | green transitions
 | fully-green
\knowledge{notion}
 | kernel weighted timed games
 | kernel weighted timed game
 | kernel \wtg 
 | kernel \wtgs 
 | kernel game
\knowledge{notion}
 | output weight function
 | output weight functions
 | output function
\knowledge{notion}
 | value iteration algorithm
\knowledge{notion}
 | regions
 | region
\knowledge{notion}
 | $\intfo 01$-region
 | $\intfo 01$-regions
\knowledge{notion}
 | \wtg with clock values in $\intfo 01$
 | $\intfo 01$-\wtgs
 | $\intfo 01$-\wtg 
 | \textbf {[0,1)-}\wtg 
\knowledge{notion}
 | \textbf {[0,1]-}\wtg 
 | $[0,1]$-\wtg 
 | $\intff 01$-\wtg 
\knowledge{notion}
 | projected output functions
 | projected output function

\knowledge{notion}
 | \textbf {region} \wtg 
 | region \wtg 
 | region \wtgs 
 | a region weighted timed game
 | region $\intfo 01$-\wtg 
 | region weighted timed game
 
\knowledge{notion}
 | almost-divergent \wtgs 
  | almost divergent \wtgs 
\knowledge{notion}
 | semi-unfolding
 | semi-unfolding construction
 | partial unfolding
 | semi-unfolded \wtg 
\knowledge{notion}
 | time-successor
 | time-successors
\knowledge{notion}
 | circular clock difference
\knowledge{notion}
 | trimmed region \wtg 
 | trimmed region $\intfo 01$-\wtg 
 | region trimmed $\intfo 01$-\wtg 
 | region trimmed \wtg
 | trimmed region kernel \wtg 
 | trimmed
 | trimmed region (kernel) \wtg  
\knowledge{notion}
 | relaxed region
 | relaxed region \wtg 

\knowledge{notion}
 | relaxed trimmed region
 | relaxed trimmed region \wtg 
 | relaxed region trimmed kernel $\intff 01$-\wtg 
 | relaxed trimmed region kernel $\intff 01$-\wtg 
  | relaxed region $\intfo 01$-\wtg 
   | relaxed trimmed \wtg 
   |relaxed
\knowledge{notion}
 | region assignment
\knowledge{notion}
 | adherence

\knowledge{notion}
 | \textbf {bisimulation}
 | bisimulation
 | (bi)simulation

\knowledge{notion}
 | Value as a function
 | Value function
\knowledge{notion}
 | $\epsilon $-neighbourhood
 | $\epsilon $-neighbours
 | $\epsilon $-neighbour
 | epsilon-neighbours

\knowledge{notion} 
 | upper clocks
  | upper clock
\knowledge{notion} 
| sub-\wtg 
 | sub-\wtgs 
 \knowledge{notion} 
  | output functions
\knowledge{notion}
 | Relaxed guards 
  | relaxed guards

\knowledge{notion} 
 | SCC
\knowledge{notion} 
 | up-driven

\knowledge{notion} 
  | \emph {urgent}
  | urgent
  | non-urgent
 
\knowledge{notion} 
 | time-boundedness
 | time bound
 | time-bounded \wtgs 
 | time-bounded \wtg 
 | time-bounded

\knowledge{notion}
 | ideal
 | ideal strategy
 
 \knowledge{notion}
 | ideal strategy invariant

\knowledge{notion}
|shortcut transition
|shortcut transitions
|shortcut

\knowledge{notion}
| min-cost reachability game
| Min-cost reachability game
| min-cost reachability games
| Min-cost reachability games

\title{Approximating the Value of Time-Bounded Weighted Timed Games}

\author{Isa {Vialard}}{Max Planck Institute for Software Systems, Saarland Informatics Campus, Germany \and \url{https://isavialard.github.io/home/} }{vialard@mpi-sws.org}{https://orcid.org/0000-0002-7261-9342}{}
\author{Joël Ouaknine}{Max Planck Institute for Software Systems, Saarland Informatics Campus, Germany}{joel@mpi-sws.org}{}{}

\authorrunning{I. Vialard, J. Ouaknine} %TODO mandatory. First: Use abbreviated first/middle names. Second (only in severe cases): Use first author plus 'et al.'

\Copyright{Isa Vialard, Joël Ouaknine} %TODO mandatory, please use full first names. LIPIcs license is "CC-BY";  http://creativecommons.org/licenses/by/3.0/

 %TODO mandatory; please add comma-separated list of keywords

\category{} %optional, e.g. invited paper

\nolinenumbers %uncomment to disable line numbering

\begin{document}

\maketitle

\begin{abstract}
The Value Problem for weighted timed games, an extension of timed automata with weights and adversity, consists in determining, given a
weighted timed game with a reachability objective and a rational
threshold, whether or not the value of the game exceeds the threshold. This
problem is known to be undecidable for weighted timed games, i.e., there is no algorithm that can compute the value exactly.
 A natural alternative is to approximate the value; however, the Approximation problem is also undecidable in general. 
 
 In this paper, we prove that the Approximation problem is decidable in $4$-\textsc{ExpTime} for \wtgs whose executions are time-bounded.
\end{abstract}
%New open problem: Approximability 2 clock
%
%
%Here WTG turn-based with $2$ clocks, positive weight. 
%\Wlog with clocks in $[0,1]$.
%Trimmed region game. Suppose that Min can reach target, so remove any Max-controlled cycle.

\section{Introduction}

Timed automata, introduced by Alur and Dill \cite{Alur94}, are a fundamental model for the verification of real-time systems. They extend finite automata with a finite number of real-valued clocks, whose values govern the enabling of transitions, while time can elapse continuously within locations. This combination of discrete and continuous dynamics gives timed automata great expressiveness, making them well suited for modeling reactive systems subject to fine-grained timing constraints.
Weighted timed games (\wtgs) are timed automata equipped with
integer weights on locations and transitions, and a game structure. The weights on transitions are instantaneous, while those on locations are rates that multiply time delays.
Locations are partitioned between player Min, who seeks to minimize the cost required to reach a target location, and player Max, who represents an uncontrollable environment.

The main problem that arises for timed games is the Value problem: given a \wtg
$\Gcal$ and a threshold 
$c$, is the value (that is, the infimum weight that a $\Minsf$'s strategy can guarantee) less than or equal to $c$. In 2015, Bouyer, Jaziri and Markey \cite{bouyer2015} showed that the Value problem is undecidable for \wtgs with three clocks or more, even under an almost-divergence~\footnote{
A \wtg is almost-divergent if any run following a cycle has cost in $\intof{-\infty}{-1}\cup \{0\} \cup \intfo{1}{+\infty}$
}
 constraint and while restricting weights to $\mathbb{N}$.
However, this problem is decidable for $1$-clock \wtgs (\cite{bouyer2006,monmege2024}).  
Recently, Guilmant, Ouaknine and Vialard \cite{two-clock-undec} showed that the Value problem remains undecidable for two-clock \wtgs, even under a time-boundedness constraint; this result closes the last major open question on the Value problem.

Faced with these undecidability results, the natural question to investigate is the Approximation problem: given a \wtg $\Gcal$, and some $\epsilon>0$, can we compute an $\epsilon$-approximation of the value?
Bouyer, Jaziri and Markey \cite{bouyer2015} gave an approximation algorithm for almost-divergent \wtgs with non-negative weights. This result was then extended to arbitrary weights (\cite{Busatto2018}). However,
Guilmant and Ouaknine \cite{guilmant2024} proved that the Approximation problem  for weighted timed games is undecidable by reduction from the promise problem on probabilistic automata. 

In this paper, we show that the Approximation problem is decidable for weighted timed games under a \emph{time-boundedness} constraint. Time-boundedness does not substantially restrict the computational power of weighted timed games, as shown by the Value undecidability result in \cite{two-clock-undec}.
Furthermore, studying timed automata under a bounded time horizon is well motivated from a practical perspective. Indeed, the execution of real-time systems is inherently limited to finite-observation windows, either because the system itself operates over a bounded duration or because verification naturally focuses on a finite segment of its execution. Moreover, the chosen time bound is typically of the same order of magnitude as the timing constraints specified in the model.

\paragraph*{Structure of the article.}
\cref{sec-bisim} proves that the value function is Lipschitz-continuous for weighted timed games in a time-bounded setting, which implies that the classical "value iteration algorithm" converges to the value function. 
\cref{sec-approx} provides an explicit upper bound on the convergence rate of the "value iteration algorithm". As a consequence, we obtain a $4$-\textsc{ExpTime} algorithm for computing arbitrarily precise approximations of the value.

Our result constitutes a significant step forward, as the strongest previously known decidability result for Approximation  (\cite{Busatto2018}) concerns almost-divergent\footnote{A \wtg is \emph{almost divergent} if every strongly connected component has either cycles of weight in $\intof{-\infty}{-1}\cup\{0\}$, or 
cycles of weight in $\{0\}\cup\intfo{1}{\infty}$.}
\wtgs. Even though almost-divergence and time-boundedness are orthogonal, the former can be seen as more constraining in practice, since the latter is naturally satisfied in real-time systems.

\noindent\textbf{Disclosure:}
AI-based tools were used to assist with writing and proof-checking, and for the computation in \cref{unfolding-N}. The authors remain fully responsible for the content of the paper.
%\delta for delay
%d for number of clocks
%t for transition
%T for set of transition
%\Tmax for time-bound

\section{Definitions}
\subsection{Weighted timed games}
Let $\mathcal{X}$ be a finite set of \textbf{clocks}. \textbf{""Clock constraints""} (or \textbf{guards}) over $\mathcal{X}$ are
	expressions of the form $x \mathrel{\bowtie} n$, where $x\in\mathcal{X}$ is a clock,
        ${\bowtie}\in\{<,\leq,=,\geq,>\}$ is a comparison symbol, and
        $n\in\mathbb{N}$ is a natural number.  We denote with $\mathcal{C}$ 
        the set of all clock constraints over $\mathcal{X}$.
A \textbf{valuation} on $\mathcal{X}$ is  a function $\nu:\mathcal{X}\to\mathbb{R}_{\geq 0}$. 
	For $\delta\in\mathbb{R}_{\geq 0}$ we denote with $\nu+\delta$ the valuation such
        that, for every clock $x \in \mathcal{X}$, $(\nu+\delta)(x)=\nu(x)+\delta$. Let
        $X \subseteq \mathcal{X}$ be a subset of all clocks. We write
        $\nu[X :=0]$ for the valuation such that, for every clock $x
        \in X$, $\nu[X :=0](x) = 0$, and $\nu[X:=0](y) = \nu(y)$ for
        all other clocks $y \notin X$.
	For $C \subseteq \mathcal{C}$ a set of clock constraints over $\mathcal{X}$, we say that
        the valuation $\nu$ \textbf{satisfies} $C$, denoted with $\nu \models
        C$, if and only if all the comparisons in $C$ hold
        when replacing each clock $x$ by its corresponding value $\nu(x)$.

\begin{definition}
	A \textbf{(turn-based) weighted timed game} is given by a
        tuple $\mathcal{G}=$\linebreak $(L_\mathsf{Min}, L_\mathsf{Max},G,\mathcal{X},T,w)$, where:
	\begin{itemize}
        \item $L_\mathsf{Min}$ and $L_\mathsf{Max}$ are the (disjoint)
          sets of \textbf{locations} belonging to Players $\mathsf{Min}$ and
          $\mathsf{Max}$ respectively; we denote with $L = L_\mathsf{Min} \cup
          L_\mathsf{Max}$ the set of all locations. (In drawings,
          locations belonging to $\mathsf{Min}$ are depicted by blue
          circles, and those belonging to $\mathsf{Max}$ are depicted
          by red squares.)
                   \item $G\subseteq L_{\mathsf{Min}}$ are the \textbf{goal locations}.
		\item $\mathcal{X}$ is a set of clocks.
		\item $T\subseteq (L \setminus G) \times2^\mathcal{C}\times 2^\mathcal{X}\times
                  L$ is a set of \textbf{(discrete) transitions}. A
                  transition $\ell \xrightarrow{C,X} \ell'$
                  enables moving from location $\ell$ to
                  location $\ell'$, provided all "clock constraints" in
                  $C$ are satisfied, and afterwards resetting all clocks in $X$
                  to zero. 
		\item $w:(L \setminus G) \cup T\to \mathbb{Z}$ is a \textbf{weight function}.
                \end{itemize}
In the above, we assume that all data (set of locations, set of
clocks, set of transitions, set of "clock constraints") are finite.             
\end{definition}

          Let $\mathcal{G}= (L_\mathsf{Min}, L_\mathsf{Max},G,\mathcal{X},T,w)$ be
          a \wtg.  A \textbf{configuration} over
          $\mathcal{G}$ is a pair $(\ell,\nu)$, where $\ell\in L$ and $\nu$
          is a valuation on $\mathcal{X}$. Let $\delta \in \mathbb{R}_{\geq 0}$ be a
          \textbf{delay} and $t = \ell \xrightarrow{C,X} \ell' \in T$
          be a discrete transition. One then has a \textbf{delayed
            transition} (or simply a \textbf{transition} if the
          context is clear)
          $(\ell,\nu) \xrightarrow{\delta,t} (\ell',\nu')$ provided that
          $\nu+\delta \models C$ and $\nu' = (\nu+\delta)[X := 0]$. Intuitively,
          control remains in location $\ell$ for $\delta$ time units, after
          which it transitions to location $\ell'$,
          resetting all the clocks in $X$ to zero in the process.  The
          \textbf{weight} of such a delayed transition is
          $\delta \cdot w(\ell) + w(t)$, taking into account both the time
          spent in $\ell$ and the weight of the discrete
          transition $t$.
          
\AP          We say that a \wtg is ""time-bounded"" if there exists $\Tmax\in\mathbb{N}$ such that every play starting in a configuration $(\ell,\bold{0})$ lasts at most $\Tmax$ time units. This property is decidable (see \cref{lem-deciding-time-bounded}).
%  and some clock $t$ that cannot be reset such that every transition verifies $t\leq T$.

          As noted in~\cite{BusattoGastonMR23},
          without loss of generality, one can assume that no
          configuration (other than those associated with goal
          locations) is deadlocked; in other words, for any location
          $\ell \in L\setminus G$ and valuation
          $\nu \in \mathbb{R}_{\geq 0}^{\mathcal{X}}$, there exists
          $\delta \in \mathbb{R}_{\geq 0}$ and $t \in T$ such that
          $(\ell,\nu) \xrightarrow{\delta,t} (\ell',\nu')$.\footnote{This
            can be achieved by adding transitions to a
            sink location for all locations controlled by $\mathsf{Min}$ and
            unguarded transitions to a goal location for the ones
            controlled by $\mathsf{Max}$. To preserve time-boundedness, we ask that these transitions be urgent (\ie they can only be taken with delay $0$). This can be enforced by adding a clock, that is reset by every transition, and only consider valuation where this clock is null. }

          Let $k\in\mathbb{N}$.  A \textbf{run} $\rho$ of length $k$ over $\mathcal{G}$ from a given configuration
          $(\ell_0,\nu_0)$ is a sequence of matching delayed
          transitions, as follows:
          \[ \rho = (\ell_0,\nu_0) \xrightarrow{\delta_0,t_0}
            (\ell_1,\nu_1) \xrightarrow{\delta_1,t_1} \cdots
            \xrightarrow{\delta_{k-1},t_{k-1}}  (\ell_k,\nu_k) \, . \]
          The \textbf{weight} of $\rho$ is the
          cumulative weight of the underlying delayed transitions:
          \[ \mathsf{weight}(\rho) = \sum_{i=0}^{k-1} (\delta_i\cdot
            w(\ell_i) + w(t_i)) \, . \]
          An infinite run $\rho$ is defined in the obvious way;
          however, since no goal location is ever reached, its weight is defined to be infinite:
          $\mathsf{weight}(\rho)=+\infty$.

          A run is \textbf{maximal} if it is either infinite or cannot be extended
          further. Thanks to our deadlock-freedom assumption, finite maximal
          runs must end in a goal location. We refer to maximal runs
          as \textbf{plays}.

          For any Player $\mathsf{P} \in
          \{\mathsf{Min},\mathsf{Max}\}$, we denote with
 $\mathcal{FR}_{\mathcal{G}}^{\mathsf{P}}$
 the collection of all non-maximal finite runs of $\mathcal{G}$ 
          ending in a location belonging to
          Player $\mathsf{P}$. A \textbf{strategy}
         for Player $\mathsf{P}$ is a mapping
           $\sigma_{\mathsf{P}} : \mathcal{FR}_{\mathcal{G}}^{\mathsf{P}}
           \to \mathbb{R}_{\geq 0} \times T$ such that for all 
           finite runs $\rho \in \mathcal{FR}_{\mathcal{G}}^{\mathsf{P}}$
           ending in configuration $(\ell,\nu)$ with
           $\ell \in L_{\mathsf{P}}$, the delayed transition
           $(\ell,\nu) \xrightarrow{\delta,t} (\ell',\nu')$ is valid,
           where $\sigma_{\mathsf{P}}(\rho)=(\delta,t)$ and $(\ell',\nu')$
           is some configuration (uniquely determined by
           $\sigma_{\mathsf{P}}(\rho)$ and $\nu$).

\AP Let us fix a starting configuration $(\ell_0,\nu_0)$, and let
$\sigma_{\mathsf{Min}}$ and $\sigma_{\mathsf{Max}}$ be strategies for
Players $\mathsf{Min}$ and $\mathsf{Max}$ respectively (one speaks of
a \emph{strategy profile}). Let us denote with
$\intro*\play_\mathcal{G}((\ell_0,\nu_0),\sigma_{\mathsf{Min}},\sigma_{\mathsf{Max}})$
the unique maximal run starting from configuration
$(\ell_0,\nu_0)$ and unfolding according to the strategy profile
$(\sigma_{\mathsf{Min}},\sigma_{\mathsf{Max}})$: in other words,
for every strict finite prefix $\rho$ of
$\play_\mathcal{G}((\ell_0,\nu_0),\sigma_{\mathsf{Min}},\sigma_{\mathsf{Max}})$
in $\mathcal{FR}_{\mathcal{G}}^{\mathsf{P}}$, the delayed transition
immediately following $\rho$ in
$\play_\mathcal{G}((\ell_0,\nu_0),\sigma_{\mathsf{Min}},\sigma_{\mathsf{Max}})$
is labelled with $\sigma_{\mathsf{P}}(\rho)$.

%Recall that finite plays have finite weights, and infinite plays have weight infinite.
Recall that the objective of Player $\mathsf{Min}$ is to reach a goal
location through a play of minimal weight. Player
$\mathsf{Max}$ has the opposite objective, trying to avoid goal
locations, and, if not possible, to maximise the cumulative weight of
any attendant play. This gives rise to the following two symmetrical definitions:
\begin{align*}
\overline{\mathsf{Val}}_\mathcal{G}(\ell_0,\nu_0) &=
  \inf_{\sigma_{\mathsf{Min}}} \left\{ \sup_{\sigma_\mathsf{Max}}
                                              \left\{
                                              \mathsf{weight}(\play_\mathcal{G}((\ell_0,\nu_0),\sigma_{\mathsf{Min}},\sigma_{\mathsf{Max}}))\right\}
  \right\} \mbox{\ and}\\
\underline{\mathsf{Val}}_\mathcal{G}(\ell_0,\nu_0) &=
  \sup_{\sigma_{\mathsf{Max}}} \left\{\inf_{\sigma_\mathsf{Min}}
                                               \left\{\mathsf{weight}(\play_\mathcal{G}((\ell_0,\nu_0),\sigma_{\mathsf{Min}},\sigma_{\mathsf{Max}}))
                                               \right\} \right\} \, .
\end{align*}

$\overline{\mathsf{Val}}_\mathcal{G}(\ell_0,\nu_0)$ represents the smallest 
possible weight that Player $\mathsf{Min}$ can possibly achieve,
starting from configuration $(\ell_0,\nu_0)$,
against optimal play from Player $\mathsf{Max}$, and conversely for
$\overline{\mathsf{Val}}_\mathcal{G}(\ell_0,\nu_0)$: the latter represents
the largest possible weight that Player $\mathsf{Max}$ can enforce,
against best play from Player $\mathsf{Min}$.\footnote{Technically
  speaking, these values may not be literally achievable; however
  given any $\varepsilon > 0$, both players are guaranteed to have
  strategies that can take them to within $\varepsilon$ of the optimal value.}
\AP As noted in~\cite{BusattoGastonMR23}, turned-based
\wtgs are \emph{determined}, and therefore
$\overline{\mathsf{Val}}_\mathcal{G}(\ell_0,\nu_0) =
\underline{\mathsf{Val}}_\mathcal{G}(\ell_0,\nu_0)$ for any starting
configuration $(\ell_0,\nu_0)$; we denote this common value with
$\intro*\Value{\mathcal{G}}(\ell_0,\nu_0)$.
Note that $\Value{\mathcal{G}}(\ell_0,\nu_0)$ can take on real numbers, or either
	of the values $-\infty$ and $+\infty$.

%\textbf{In the remainder of this paper, every weighted timed game is turn-based, with non-negative weights, of value in $\mathbb{R}$}.

\begin{definition}[Value problem]
    The \emph{""Value problem""} asks, given a \wtg $\Gcal$, an initial configuration
    $(\ell,\nu)$ and a threshold $c\in\mathbb{Z}$, whether
    $
        \Value{\Gcal}(\ell,\nu)\leq c.
    $
\end{definition}
\begin{definition}[Approximation problem]
    The \emph{""Approximation problem""} asks, given a \wtg $\Gcal$, an initial
    configuration $(\ell,\nu)$ and a precision $\epsilon>0$, to compute a value
    $V\in\mathbb{R}\cup\{-\infty,+\infty\}$ such that
    $
        \left|V-\Value{\Gcal}(\ell,\nu)\right|\leq\epsilon.
    $ if $V\in R$, and $V=\Value{\Gcal}$ otherwise.
\end{definition}

\subsection{Min-cost reachability games}

\begin{definition}
A ""min-cost reachability game"" is given by a tuple $\Mcal=(L_\Minsf, L_\Maxsf, G, T, w)$ where
\begin{itemize}
\item $L_\mathsf{Min}$ and $L_\mathsf{Max}$ are the (disjoint)
          sets of \textbf{locations} belonging to Players $\mathsf{Min}$ and
          $\mathsf{Max}$ respectively; we denote with $L = L_\mathsf{Min} \cup
          L_\mathsf{Max}$ the set of all locations. (In drawings,
          locations belonging to $\mathsf{Min}$ are depicted by blue
          circles, and those belonging to $\mathsf{Max}$ are depicted
          by red squares.)
          \item $G\subseteq L_{\mathsf{Min}}$ are the \textbf{goal locations}.
		\item $T\subseteq (L \setminus G) \times
                  L$ is a set of \textbf{(discrete) transitions}.
		\item $w: T\to \mathbb{Z}$ is a \textbf{weight function}.
\end{itemize}
\end{definition}
"Min-cost reachability games" can be seen as a subclass of \wtgs, where the "Value problem" is decidable (\cite{BGHM16}). 

Observe that, for
 any \wtg, one can obtain a "min-cost reachability game" by removing clocks, clock constraints and location weights.

\subsection{Region construction and $[0,1)$-\wtgs}

%
%\subsection{Restraining Clock Values to [0,1)}
%\label{sec-zero-one}
\begin{definition}
A ""$\intfo01$-\wtg"" is a \wtg with the global condition $x\leq 1$ on every clock $x$, such that every clock reaching $1$ is immediately reset.
\end{definition}

\begin{lemma}[\cite{vialard2025decidingvaluetwoclocknonzeno}]
Any weighted timed game with bounded clocks can be turned into  an equivalent "$\intfo01$-\wtg".
\end{lemma}

The main idea of the construction is to use locations to hold the information of the integer parts of clock valuations. Since every clock is bounded, this information is finite.

The main interest in dealing with "$\intfo01$-\wtgs" is that it allows us to present the well-known notions of  "regions" and "region \wtgs"
with simpler notations.

\begin{definition}
	Let $\Xcal$ be a set of clocks.
%	 in a "$\intff01$-\wtg". 
%	Let $M,N\in\mathbb N$ be natural numbers. $M$ is the \textbf{bound} and $N$ is the
%	\textbf{granularity}. 
	A \textbf{""region""} over $\Xcal$ is a tuple $r=(X_0,\dots,X_p,X_{=1})$ such that $X_i\neq \emptyset$ for all $1\leq i \leq p$, and
$\{X_0,\dots,X_p,X_{=1}\}$ is a partition of $\Xcal$: $\Xcal=\biguplus_{i=0}^pX_i \uplus X_{=1}$.

	We denote with $\intro*\Reg{\Xcal}$ the set of regions over $\Xcal$.
	A valuation $\nu$ is said to \textbf{belong} to region $r$, denoted with $\nu\in r$, whenever
	\begin{itemize}
		\item $\forall x\in\Xcal$, $\nu(x)=0\Leftrightarrow x\in X_0$,
		\item $\forall x\in\Xcal$, $\nu(x)=1\Leftrightarrow x\in X_{=1}$,
		\item $\forall x,y\in\Xcal$, $0< \nu(x)< \nu(y)<1\Leftrightarrow \exists i,j\in\{1,\dots, p\}\text{ s.t. } i<j\wedge x\in X_i\wedge y\in X_j$.
%		\item $\forall x,y\in\Xcal$, $\ \nu(x)= \nu(y)<1\Leftrightarrow \exists i\in\{0,\dots, p\}\text{ s.t. } x,y\in X_i$.
	\end{itemize}

For $r=(X_0,\dots, X_p, X_{=1})$ a "region" and $X$ a non-empty subset of $\Xcal$, we denote with $r[X:=0]$ the "region" that contains $\nu[X:=0]$ for all $\nu\in r$.
%Formally
%$(X_0\cup X,X'_1\setminus X,\dots, X'_{p'}\setminus X, X_{=1}\setminus X)$
%where $(X'_1,\dots, X'_{p'})$ is obtained from $(X_1,\dots, X_{p})$ by removing any $X_i\subseteq X$.
% In other words, $r[X:=0]$ is the "region"   
 \AP We say that $r$ is \textbf{""$\intfo01$-region""} when $X_{=1}=\emptyset$, in which case we abuse notation and write $r$ as  $(X_0,\dots, X_p)$. We denote with $\intro*\Reglt{\Xcal}$ the set of "$\intfo01$-regions" over $\Xcal$.
Clocks in $X_p$ are called ""upper clocks"" of $r$. 

\AP The \textbf{""time-successors""} of a "region"  $r$ are the "regions" that can be visited from $r$ as time progresses. More formally, a "time-successor" of $r=(X_0,\dots, X_p)\in \Reglt{\Xcal}$  is a "region" $r'=(X'_0,\dots, X'_{p'}, X'_{=1})$ such that either $r'=r$; or $X_0'=\emptyset$ and 
$X'_i=X_{i-1}$ for $1\leq i \leq p$, and either $X_p=X'_{=1}$ or $X_p=X'_{p+1}$ (and  then $X'_{=1}=\emptyset$).

\end{definition}

%We often abuse notation and write $r$ for the set of valuations satisfying $\nu\in r$.

\begin{definition}[Region \wtg \cite{Alur94}]
Let $\mathcal G = (L_\mathsf{Min}, L_\mathsf{Max},G,\mathcal{X},T,w)$ be a "$\intfo01$-\wtg".
The \textbf{region construction} $\intro*\RC{\mathcal G}$ of $\Gcal$ is the "$\intfo01$-\wtg" $\RC{\Gcal}=(L'_\mathsf{Min}, L'_\mathsf{Max},G',\mathcal{X},T',w')$ with
\begin{itemize}
	\item $L'_\Player=L_\Player\times\Reglt{\Xcal}$
	for $\mathsf P\in\{\Minsf,\Maxsf\}$.
	\item $G'=G\times\Reglt{\Xcal}$.
	\item For every $r=(X_0,\dots X_p)\in \Reglt{\Xcal}$, for every $r'=(X'_0,\dots, X'_{p'}, X'_{=1})\in \Reg{\Xcal}$ a "time-successor" of $r$,
	if $t:\ell \xrightarrow{C,X} \ell' \in T$ then
	$t_r:(\ell,r) \xrightarrow{C\cup C(r'),X} (\ell',r'[X:=0])\in T'$,
with $$C(r')=\setof{(x=0)}{x\in X'_0}\;\cup\; \setof{(x=1)}{x\in X'_{=1}}\;\cup\; \setof{(0<x<1)}{x\in X'_i, 1\leq i\leq p}.$$
\item for all $t\in T$, $r\in \Reglt{\Xcal}$, $w'(t_r)=w(t)$.
	\item For $\ell\in L_{\Minsf}\cup L_{\Maxsf}$ and $r\in\Reglt{\Xcal}$, $w'(\ell,r)=w(\ell)$.
\end{itemize}
For any $\ell'=(\ell,r)\in L'$, we denote $r$ with $\reg{\ell'}$.
\end{definition}
Applying the $\intfo01$-construction then the region construction yields the same location complexity as applying directly the classical region construction
as defined in \cite{Alur94}.

%%%%%%%%%%%%%%%%%%% Past trim here

\AP Since the previous transformations preserve the value, we can now assume that any \wtg $\Gcal$ is \textbf{a "region weighted timed game"}, defined as in \cite{vialard2025decidingvaluetwoclocknonzeno}:
\begin{definition}
A ""region \wtg"" is a \wtg with a ""region assignment"" $\intro*\reg\cdot : L\to\Reglt{\Xcal}$, such that
%	\item there is some initial location $(\ell_0,r_0)$.
%	\item for every location $(\ell,r)$, for every $\nu_0\in \tilde{r_0}$, for any path $\rho:
%	((\ell_0,r_0),\nu_0)\dots \rightarrow ((\ell,r),\nu) \rightarrow \dots$, $\nu\in \tilde{r}$.
for any transition $t:\ell\xrightarrow{C,X} \ell'$, the valuations $\nu+\delta$ with $\nu\in \reg{\ell}$ and $\delta\geq 0$ that satisfy $C$ are contained in a unique "region" $r$, such that $r[X:=0]=\reg{\ell'}$. Furthermore, for any initial configuration $(i,\nu)$, we require $\nu\in \reg{i}$.
\end{definition}

	Obviously $\RC{\Gcal}$ is a "region \wtg " for any "$\intfo01$-\wtg" $\Gcal$.

%	Let us now ``trim'' $\RC{\Gcal}$, \ie delete every useless transition, and every useless guard on transitions:

	\begin{definition}[Trimmed region \wtg]
	\label{def-trimmed}
	A  region $\intfo01$-\wtg $\Gcal$ is \textbf{""trimmed""} if for any transition $t:\ell\xrightarrow{C,X} \ell'$  and any "region" $r=\reg{\ell}$ in $\Gcal$,
	\begin{itemize}
		\item for any valuation $\nu\in r$ there exist some $\delta\geq0$ such that $\nu+\delta\models C$.
		\item for any $c\in C$, there exists a valuation $\nu\in r$ and some $\delta\geq 0$ such that $\nu+\delta$ is in $\intff01^\Xcal$ and
		$\nu+\delta\not\models c$.
	\end{itemize}
\end{definition}

In other words, trimming a "region \wtg" means deleting inaccessible transitions from any tuple location-region, and removing unnecessary guards (the ones that are always verified from the "region"). Removing useless transitions and guards can always be done from any "region \wtg" without change in value.
\begin{remark}
\label{rem-trim}
For any finite path $\pi$ in a "trimmed region \wtg" starting in a location $\ell$, for any $\nu\in \reg{\ell}$, there exists a valid run $\rho$ following $\pi$ starting from $(\ell,\nu)$.
\end{remark}

Since every "$\intfo01$-\wtg" $\Gcal$ is equivalent to the "region $\intfo01$-\wtg" $\RC{\Gcal}$, and every "region \wtg" is equivalent to a "trimmed region \wtg", we can always assume $\Gcal$ to be a "trimmed region \wtg".

%Observe that the region construction essentially preserves "time-boundedness".
% if in $\Gcal$ a clock $t$ is bounded by some time bound $T$, and cannot be reset, then in $\RC{\Gcal}$ we know that $t$ is reset at most $T$ times and only when $t=1$. Thus any play lasts at most $T$ time units (after the $T$-th reset of $t$, no positive delay can be taken).
%We say that a trimmed region $\intfo01$-\wtg has a "time bound" if it is
%obtained by applying the region-then-trim construction to a "time-bounded" \wtg. 

\begin{remark}[Finite value assumption.]
\label{rem-finite-V} 

Without loss of generality, we can assume that the value function of $\Gcal$ is in $\mathbb{R}$. Indeed, 
since reachability is decidable in timed games, the set of locations where the value function is $+\infty$ can be computed in polynomial time in $|\RC{\Gcal}|$.

However, the set of $-\infty$ location is undecidable in the general case (\cite{BusattoGastonMR23}, Proposition 3.12). Fortunately, because of the "time-bounded" setting, a $-\infty$ value can only be caused by cycles with negative transition weights. Thus, if a configuration $(\ell,\nu)$ with $\nu\in\reg{\ell}$ has infinite value, then any valuation in $\{\ell\}\times \reg{\ell}$ has infinite value too. The $-\infty$ locations in a trimmed region time-bounded \wtg are thus exactly $-\infty$ locations of the corresponding "min-cost reachability game", which can be computed in polynomial time (if the transition weights are given in unary), according to \cite{BGHM16}.

Either the initial location is a $\pm\infty$ location, or we can remove every location of infinite value without modifying the value in the remaining locations, since $\Minsf$ (resp. $\Maxsf$) can avoid any location of value $+\infty$ (resp. $-\infty$).
%\todo{chck and cite}
\end{remark}

%\begin{remark}[$0$-delay cycles]
%\label{rem-nodelay}
%Without loss of generality, we consider only pairs of $\Minsf$ and $\Maxsf$ strategies that cannot produce a play with a $0$-delay cycle in a region game. Indeed, if this cycle has negative transition weight, then $\Minsf$ cannot stop $\Maxsf$ from leaving this cycle and conversely, otherwise this would yield a $-\infty$ or $+\infty$ Value, respectively. Thus one of the player can and should leave this cycle. If this cycle has non-negative transition weight, then such a cycle holds no interest for $\Minsf$. Indeed, in a region game, following a $0$-delay cycle does not change the clock valuation Since there is no location of value $+\infty$ (see \cref{rem-finite-V}), then $\Maxsf$ cannot stop $\Minsf$ from leaving this cycle.
%\end{remark}
%%Cette remarque sert à ce que la stratégie idéale ne nous fasse pas tourner en rond
%%Mais en fait on veut vraiment trimmer

\begin{definition}[Sub-\wtgs]
A ""sub-\wtg"" is a region \wtg with ""output functions"" assigned  to each location in $G$. These output functions are piecewise linear functions defined on the region of the associated location.
\end{definition}
\AP
This notion allows us to study independently a part of a \wtg if we know the value function outside (or an approximation); this value function is given to the "sub-\wtg" as its "output functions".
We call an ""output edge"" any edge in a "sub-\wtg" leading to a goal location.

\subsection{Value iteration algorithm}

%\begin{theorem}
%\label{via-tends}
%The value iteration algorithm tends toward the value for  \wtg with non-negative weights and any number of clocks.
%\end{theorem}

\begin{definition}[Value iteration algorithm]
Let $\Gcal=(L_{\Minsf},L_{\Maxsf},G,\Xcal,T,w)$ be a "region \wtg".
The \textbf{""value iteration algorithm""} builds, for each location $\ell$ and for all $k\geq 0$, a function $\intro*\opt{\ell}{k}: \reg{\ell} \to \Rbb\cup\{+\infty\}$ such that
$\opt{\ell}{k}(\nu)$ is the value of the game starting in $\ell$ with clock valuation $\nu$, when Min is restricted to strategies reaching G within k steps.
The $\opt{}{}$ functions are built inductively:
\begin{itemize}
\item If $\ell \in G$, $\opt{\ell}{k}$ is the constant $0$ function for any $k\in \mathbb{N}$, or the output function associated to $\ell$ if $\Gcal$ is a "sub-\wtg".
\item If $\ell\not\in G$, $\opt{\ell}{0}$ is the constant $+\infty$ function.
\item for any $k\in \Nat$, for any $\ell\not\in G$, $\opt{\ell}{k+1}$ is obtained from the $\opt{}{}$ functions at step $k$: if $\ell$ belongs to $\Minsf$ (resp.~$\Maxsf$), then for all $\nu\in\reg{\ell}$,
\[\opt{\ell}{k+1}(\nu)=
\underset{\text{(resp.~$\sup$) }}{\inf}  \enstq{\opt{\ell'}{k}((\nu+\delta)[X:=0])+\delta \cdot w(\ell)+w(t)}
{\begin{array}{c}
t:\ell \overset{C,X}{\rightarrow} \ell'\in T,\\
 \delta\geq 0, \nu+\delta \models C
\end{array}
}.\]
\end{itemize}
(Because of the no-deadlock assumption, this set is not empty.)
%We denote with $\Fcal$ the operator which computes $\opt{}{k+1}$ from $\opt{}{k}$.
\end{definition}

Note that $\opt{\ell}{k}(\nu)=\Value{\Gcal}^{\leq k}(\ell,\nu)$ the value of the game starting from configuration $(\ell,\nu)$ when $\Minsf$ must reach $G$ in at most $k$ transitions.
Naturally, the value iteration operator is monotonous ($\opt{\ell}{k+1}(\nu)\leq \opt{\ell}{k}(\nu)$ for all valuation $\nu$). If there exists $k$ such that, for all locations $\ell$, $\opt{\ell}{k+1}= \opt{\ell}{k}$, then the "value iteration algorithm" terminates and $\opt{\ell}{k}(\nu)=\Value{\Gcal}(\ell,\nu)$. Consequently, the value of the \wtg is obtained even when considering plays of length at most $k$.

% is the value of the game started in $\ell$ with clock valuation $\nu$. 
Unfortunately, termination is not guaranteed (otherwise the "Value problem" would be decidable), except in some special cases:
one-clock \wtgs with non-negative weights (\cite{bouyer2006}), or
two-clock almost non-Zeno \wtgs (\cite{vialard2025decidingvaluetwoclocknonzeno}).
%Indeed, Guilmant, Ouaknine and Vialard \cite{two-clock-undec} have shown that the value iteration does not terminate in general on \wtgs with two or more clocks.

 Furthermore the "value iteration algorithm" might not even tend toward the value function:

\begin{figure}
\begin{center}
\includegraphics[scale=1]{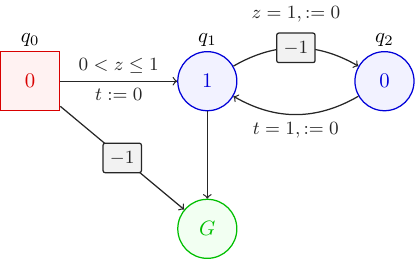}
\end{center}
 \caption{\Wtg of value $-1$, where $\opt{}{k}(0,0)=0$ for all $k\in\Nat$.}
  \label{fig-ctrex}
\end{figure}

\begin{example}
In \cref{fig-ctrex}, starting in $q_0$, $\Maxsf$ can either end the game with cost $-1$, or wait for $\epsilon>0$ time units and move to $q_1$. From $q_1$, $\Minsf$ has access to a cycle of weight $-\epsilon$, and a direct transition to goal location. Thus $\Value{}(q_1,(0,\epsilon))=-\infty$ if $\epsilon>0$, therefore $\Value{}(q_0,(0,0))=\max \{-\infty,-1\}=-1$. However, $\opt{q_1}{k}(0,\epsilon)=- \frac{k-1}{2}\epsilon$. Thus, $\opt{q_0}{k+1}(0,0)=\sup_{0<\epsilon\leq 1} \opt{q_1}{k}(0,\epsilon)=0$.
\end{example}

Monmege, Parreaux and Reynier \cite{monmege2024}(Theorem 2.4) showed that the "value iteration algorithm" tends toward the value for one-clock weighted timed game. In \cref{sec-bisim} we will show that "time-bounded \wtgs" verify this property too, and give an upper bound on the convergence speed of the "value iteration algorithm" in \cref{sec-approx}.

%According to \cite{monmege2024}, if the $\opt{}{}$ functions uniformly converge over each region, then they tend toward their greatest fixpoint, which is equal to the value.

\section{Continuity and slope of the Value function in bounded time}
\label{sec-bisim}

%\begin{lemma}
%Maximum slope in absolute value of the Value function is $\smax\leq 2T\wmax d $.
%\end{lemma}
%NB: le 2 est là car $||\nu_1-\nu_2||_\infty \leq \epsilon$ implies $2\epsilon$-close par glissières.

%To estimate how much deviating from the ideal strategy can cost, we need the following result.
\begin{restatable}[Lipschitz-continuity of the Value function]{theorem}{LipTB}
\label{thm-Lip-TB}
Let $\Gcal$ be a "region $\intfo01$-\wtg" with a "time bound" $\Tmax$. Let $\wmax$ be the maximal difference between two elements of $\{0\}\cup\{w(\ell):\ell\in L\}$, and let $d$ be the number of clocks. Then for any location $\ell$, the Value function $\nu\mapsto \Value\Gcal(\ell,\nu)$ is Lipschitz-continuous, of slope at most $2\Tmax\wmax d $.
%Moreover, if $\Gcal$ is a "sub-\wtg" with "output functions" of slope at most $s$, then Value function $\nu\mapsto \Value\Gcal(\ell,\nu)$ is is Lipschitz-continuous, of slope at most $2T\wmax d +s $.
\end{restatable}
\begin{proof}[Proof idea]

For any $\epsilon>0$,  consider a pair configurations $c$ and $c'$ of distance $\distinf{c-c'}\leq \epsilon$. W.l.o.g., let $\Value{}(c)\leq \Value{}(c')$.
For any strategies $\sigma_\Minsf$ starting from $c$, and $\sigma_\Maxsf$ starting from $c'$, imagine that you can build two strategies  $\tilde{\sigma}_\Minsf$ and $\tilde{\sigma}_\Maxsf$ that mimic ${\sigma}_\Minsf,{\sigma}_\Maxsf$  to ensure that $\weight(\play(\sigma_\Minsf,\tilde{\sigma}_\Maxsf, c))-\weight(\play(\tilde{\sigma}_\Minsf,\sigma_\Maxsf,c'))|\leq f(\epsilon)$
for some function $f$.
Then 
\begin{align*}
\Value{}(c)&\geq \inf_{\sigma_\Minsf} \sup_{\sigma_\Maxsf}
\weight(\play(\sigma_\Minsf,\tilde{\sigma}_\Maxsf, c))\\
&\geq \inf_{\sigma_\Minsf} \sup_{\sigma_\Maxsf} 
\weight(\play(\tilde{\sigma}_\Minsf,\sigma_\Maxsf,c')) - f(\epsilon)\\
&\geq \Value{}(c')-f(\epsilon)
\end{align*}

With a bisimulation argument, we will show in appendix how to construct the mimic strategies $\tilde{\sigma}_\Minsf$ and $\tilde{\sigma}_\Maxsf$, taking advantage of time-boundedness to ensure $f(\epsilon)=2\Tmax\wmax d\epsilon$.
\end{proof}

Note that Lipschitz-continuity suffices to prove that the value iteration algorithm converges toward the Value, using results from \cite{monmege2024} (Lemma 8.4 and Proposition 8.7). We will not go into details, as the upper bound on the convergence rate provided in \cref{general-rate} implies convergence.

As a direct consequence of \cref{thm-Lip-TB}, we can relax any strict guard of a trimmed "time-bounded \wtg" (trimming is the key to avoiding pathological behaviours):

\begin{restatable}[Relaxing guards]{corollary}{correlax}
\label{cor-relax}
Let $\Gcal$ be a "trimmed region \wtg" with a "time bound". Then relaxing every strict guard in $\Gcal$ into a strict or equal guard does not change the value.
\end{restatable}
\begin{proof}[Proof idea]
Intuitively, an optimal strategy for either player facing a guard $x>0$ or $x<1$ is sometimes to wait until $x$ is arbitrarily close to the guard.  In a "trimmed region \wtg", being allowed to pick a delay such that 
$x=0$, or $x=1$ does not open more possibilities than being arbitrarily close. Therefore, relaxing guards does not change the Value function, it just extends the Value function continuously to the closure of regions.

See full proof in appendix.
\end{proof}

\begin{definition}
A ""relaxed region $\intfo01$-\wtg""
is a \wtg with a "region assignment" $\reg\cdot : L\to\Reglt{\Xcal}$, such that
%	\item there is some initial location $(\ell_0,r_0)$.
%	\item for every location $(\ell,r)$, for every $\nu_0\in \tilde{r_0}$, for any path $\rho:
%	((\ell_0,r_0),\nu_0)\dots \rightarrow ((\ell,r),\nu) \rightarrow \dots$, $\nu\in \tilde{r}$.
for any transition $t:\ell\xrightarrow{C,X} \ell'$,
$C$ does not contain strict guards, and all
 the valuations $\nu+\delta\models C$ with $\nu$ in the closure of  $\reg{\ell}$ and $\delta\geq 0$ are contained in the closure of a unique "region" $r$, such that $r[X:=0]=\reg{\ell'}$. Furthermore, we consider only  initial configurations $(\ell_0,\nu)$ where $\nu$ is in the closure of $\reg{\ell_0}$.
\end{definition}

%Relaxing preserves "time-boundedness".

%\begin{theorem}
%\label{thm-via-tends}
%For "time-bounded \wtgs", the value iteration algorithm converges toward the Value.
%\end{theorem}
%\begin{proof}
%We know that the sequence of $\opt{}{}$ functions is monotonous. Furthermore, according to \cref{thm-Lip-TB}, every $\opt{k}{}$ function is Lipschitz-continuous of slope at most $2T\wmax d$,  thus the $\opt{}{}$ functions converges toward a continuous function. Therefore,
%according to Dini's theorem, the sequence  uniformly converges over each relaxed region. Thus, according to Monmege et. al.\cite{monmege2024}(Lemma 8.4), the sequence of $\opt{}{}$ functions tends
%toward its greatest fixpoint. Furthermore, Proposition 8.7 of \cite{monmege2024} states that this greatest fixpoint is equal to the value \footnote{The proof of this proposition is not specific to one-clock \wtgs.
%We won't go into details, as \cref{general-rate}
%implies \cref{thm-via-tends} without depending on it.
%}.
%\end{proof}

%\input{parts/neg-weight}
\section{Convergence speed of value iteration}
\label{sec-approx}

\begin{restatable}[Eliminitating transition weights from cycles]{lemma}{nonegcycle}
\label{lem-no-neg-cycle}
Let $\Gcal$ be a "trimmed region \wtg" with a "time bound" $\Tmax\in \Nbb$ and $n$ locations, such that every location has finite value. Then there exists an equivalent time-bounded \wtg $\Gcal'$ such that any transition that does not reach a goal location has null weight,
 with $k\cdot n$ locations, 
where $k=4(\Tmax\wmax+nW) + 1$.
with $W=\max\setof{|w(t)|}{t\in T}$, and with $\intro*\wmax$ the maximal difference between two elements of $\{0\}\cup\{w(\ell):\ell\in L\}$.
\end{restatable}
\begin{proof}[Proof idea.]
See full proof in Appendix. We show that, when considering almost-optimal strategies for $\Minsf$ and $\Maxsf$, the accumulated transition weight along the produced play is contained between $-2(\Tmax\wmax+nW)$ and $2(\Tmax\wmax+nW)$. We can then store this information in the set of locations, and pay the cumulated transition weight when moving to a goal location.
\end{proof}

Let $\Gcal$ be a \wtg with a "time bound" $\Tmax$ and $d$ clocks, and let $\tilde{\Gcal}$ be the \wtg obtained from $\Gcal$ by applying the $\intfo01$-region construction, then trimming, removing infinite-value locations, pushing transition weights into outgoing transitions, and relaxing.

\begin{remark}
\label{state-complexity}
If $\Gcal$ has $m$ locations, then $\tilde{\Gcal}$ has at most
$m'\cdot (4(m'W+\wmax \Tmax)+1)$ where $m'=m\cdot d!\cdot (2\Tmax+1)^{d}$. 

Indeed, 
applying the $\intfo01$-region construction multiplies the number of location by  $ d!\cdot(2\Tmax+1)^{d}$ (the number of regions). Trimming and removing infinite-value locations reduce the number of locations connected to the initial location. Relaxing leaves the number of location unchanged. Eliminating transition weights thus gives us at most $m'\cdot (4(m'W+\wmax \Tmax)+1)$ locations according to \cref{lem-no-neg-cycle}.
\end{remark}

\AP Let $\intro*\slopemax=2\Tmax\wmax d$.
According to \cref{thm-Lip-TB}, $\slopemax$ is the bound on the slope of the value function.
Note that $\slopemax$ also bounds the slope of the $\opt{}{}$ functions produced by the "value iteration algorithm"~\footnote{From $\Gcal$, one can create a \wtg where the initial location has value function $\opt{\ell}{k}$ for any $\ell$ and $k$, with same $\wmax$ and time bound $\Tmax$.}.
%Note that if $\slopemax=0$ then the value iteration terminates in at most $\Rcal(\Gcal)$ steps. \todo{No. What we can do is take the underlying min-cost reachability game, value iterate it until it finishes, that gives us a value that we use for both smax=0 and general case: it is basically the time it takes for the let's allow this 0-cycle strategy to become useless}
%\begin{description}
%\item[""Relaxed guards"".] Because of \cref{cor-relax}, we can assume that $\tilde{\Gcal}$ is a "relaxed region $\intfo01$-\wtg".
%\end{description}

%\AP Let $\intro*\MM= (2n-1)Wn+n $, where \(n=|R(G)|\) and \(W=\max_{t\in\mathcal T}|w(t)|\).
%According to \cite{BGHM15}, the "value iteration algorithm" terminates in $\MM$ steps for the  "min-cost reachability game" obtained by retaining the region locations and the transition weights of $\tilde{\Gcal}$. 

The main theorem of this section is the following:
\begin{theorem}
\label{general-rate}
Let $n$ be the number of locations of $\tilde{\Gcal}$. 
For any $0<\epsilon \leq n^2$, let $\epsilon_n=\frac{\epsilon}{ n^2}$.
We define a function $N$ inductively as
$$\begin{cases}
\intro*\Nind{\eta,0,m}=m,\\
\Nind{\eta,n+1,m}=
(K_\eta+1)\cdot \Nind{\frac\eta{K_\eta+1},n,m}
\end{cases}
$$

with $K_\eta=d \cdot \left(\left\lfloor \frac{2\slopemax}{\eta}\right\rfloor +1\right)$.

Now let
$N'_\epsilon= \Nind{\epsilon_n, n, n}\cdot n$.
%T \cdot \left( \frac{2T^2\wmax d}{\epsilon}\right)^d$.
 Then $\opt{}{N'_\epsilon}$ is an $\epsilon$-approximation of $\Value{\tilde{\Gcal}}$.
 
\end{theorem}

\begin{remark}
The transformations from $\Gcal$ to $\Gcal'$  preserve time-boundedness. Furthermore, any finite play $\rho$ in $\Gcal$ can be translated into a play $\tilde{\rho}$ in $\tilde{\Gcal}$, such that $|\rho|\leq |\tilde{\rho}|$. Therefore, the "value iteration algorithm" iterated $N'_\epsilon$ on $\Gcal$  gives an $\epsilon$-approximation of $\Value{\Gcal}$ on any finite value location of $\Gcal$.
\end{remark}

\noindent{\large{\textbf{Case $\slopemax=0$:}}
Since \(d\geq 1\), either \(\Tmax=0\), in which case no positive delay can occur, or \(w_{\max}=0\), in which case all location weights are zero. Hence, in both cases, delays contribute no cost. Moreover, since the region game is trimmed, every outgoing transition is enabled from every valuation of its source region after a suitable delay. It follows inductively that every function \(\operatorname{opt}^{\ell}_k\) is constant over \(\operatorname{reg}(\ell)\), and that value iteration coincides with value iteration in the "min-cost reachability game" obtained by retaining region locations and transition weights. The only non-zero weights in this "min-cost reachability game" are on outgoing edges, thus
the "value iteration algorithm" stabilises within
$n\leq N_\epsilon'$.

We assume \(\slopemax>0\) henceforth.

%From now on, let us assume that $\slopemax> 0$.
%\todo{basically I want to say that for all k bigger than M, max will want to escape negative cycle. I think I can put it in the cycle weight remark? After all I did assume no cycle weight before}

\subsection{Zooming in on a strongly connected component}

\AP
Let $\Gcal'$ be a "sub-\wtg" of $\tilde{\Gcal}$ made of a single strongly connected component (""SCC"").

We assume that $\Gcal'$ respects the following properties:
\begin{description}
\item[Same $\slopemax$ as $\tilde{\Gcal}$.] 
%We assume that the output functions given to $\Gcal'$ are an over-approximation of the value function of $\Gcal$.  
We assume that the value function of $\Gcal'$, and any over-approximation function produced by the "value iteration algorithm", are Lipschitz-continuous of slope at most $\slopemax$.
\item[$\Gcal'$ is up-driven.] We say that a \wtg or "sub-\wtg" is ""up-driven"" if non-null upper clocks can only be reset by "output edges".
\end{description}

%$\Gcal'$ can be described as a sequence of $\Tmax$ identical \wtgs linked together by transition $t=1,t:=0$.

\begin{restatable}[Upper bound on the convergence speed of the "value iteration algorithm" in an "SCC"]{theorem}{VIArateSCC}
\label{VIA-convergence-rate}
Let $n$ be the number of locations in $\Gcal$.
For any $\epsilon>0$, let $\epsilon_n=\frac{\epsilon}{n}$.

% Let $\eta=\frac{\epsilon_n}{slopemax\cdot d}$.
%  and let $\xi=\frac\eta2+\frac{\epsilon_n}{\Delta_w}$.

%Let $\xi= slopemax \cdot (d+1)$.
%We define a function $N$ inductively as
%$$\begin{cases}
%\intro*\Nind{\eta,0,m}=m,\\
%\Nind{\eta,n+1,m}=
%(K_\eta+1)\cdot \Nind{\frac\eta{K_\eta+1},n,m}
%\end{cases}
%$$
%
%with $K_\eta=d \cdot \left(\left\lfloor \frac{2\slopemax}{\eta}\right\rfloor +1\right)$

%K(\eta,n)=
%\xi\cdot\lceil\frac1\eta \rceil\cdot( K(\frac{\eta^2}{\xi},n-1)+1).

 Let $N_\epsilon= \Nind{\epsilon_n,n,n}$.
%T \cdot \left( \frac{2T^2\wmax d}{\epsilon}\right)^d$.
 Then $\opt{}{N_\epsilon}$ is a $\epsilon$-approximation of the value function of ${\Gcal'}$.

\end{restatable}

%
%\textcolor{red}{$\epsilon_n$ is the mistake allotted to any location in a bead. The $n$ argument is the number of not-instantaneous locations in a node. The $\eta$ argument is how much cost error I accept to unfold on a least-weight location. The purpose of $\xi$ is to convert this cost error into a clock error. The conversion rate is different considering who the least-weight location belongs to. Thus we spend time on a same least-weight location at most $\xi\cdot\lceil\frac1\eta\rceil$. When we unfold, we have to share the error for one location between all copies. When we reach $n=0$ then we have at most $n$ immediate locations, thus we can do at most $n-1$ steps.   }

\begin{proof}[Proof idea] See full proof in appendix.

%However, if we define alternatively the cost of an infinite play as an infinite sum of transition costs,
%\todo{ça va pas}
% then the value induced by the ideal
%strategy is equal to the value of the game by definition. This unusual definition of costs makes sense since, for every $\epsilon>0$, there exists a
%strategy that follows the ideal strategy for an arbitrarily number of steps before eventually
%leaving every such cycle, and whose value differs from $\Value{\Gcal'}$ by at most
%$\epsilon$. Equivalently, the value of the ideal strategy is the infimum of the values of
%these approximating strategies.

Our objective is to show that, for every $\epsilon>0$, there exists a strategy for
$\Minsf$ whose value is $\epsilon$-close to $\Value{\Gcal'}$ and whose plays have length
at most $N_\epsilon$. The existence of such a strategy is sufficient to prove
\cref{VIA-convergence-rate}; an explicit construction is not required.

The proof proceeds by progressively restricting behaviours that may occur arbitrarily
many times under some fixed strategy $\sigma_\Minsf$. Consider an event $A$ that can occur infinitely
often. We allocate an error budget $\eta>0$ and define a modified strategy $\sigma^\eta_\Minsf$ so that $A$ may
occur at most $k$ times in plays produced by $\sigma^\eta_\Minsf$, in a way that ensures a cost of at most $\sup_{\sigma_\Maxsf} \weight(\play(\sigma_\Minsf,\sigma_\Maxsf))+ \eta$. 
This modified strategy can be seen as a strategy on the \wtg obtained
by unfolding ${\Gcal'}$ into a string of $k+1$ copies, where
each occurrence of $A$ moves the play from one copy to the next. The same reasoning is
then applied recursively inside each SCC of the new \wtg, for a different event. Once folded back, the resulting strategy is a
strategy of the original \wtg.

In our setting, the event $A$ is the following: one of the players spends a positive
amount of time in a location of minimum weight. In Appendix, we show how to
construct a $\sigma^\eta_\Minsf$ strategy that produces a play that remains $\eta'$-close, for some $\eta'$ depending on $\eta$, to a
play produced by $\sigma_\Minsf$, while ensuring that if $\Maxsf$ leaves the SCC
through an output edge, then the resulting weight differs from the value by at most $\eta$.
This construction differs depending on whether
the location of minimum weight considered belongs to $\Minsf$ or $\Maxsf$. 
In both cases, this strategy produces a play where, for any cycle where time is spent in a location of minimal weight, at least $\eta'$ time elapses along the cycle.
Since at most one time unit can be spent in an "up-driven" $\intfo01$-\wtg, this suffices to bound the occurrences of event $A$.

Initially, the error budget $\epsilon$ is distributed among the locations of ${\Gcal'}$,
giving each location an allotted budget $\eta_n$. After unfolding around one location,
every location is duplicated into $K_\eta+1$ copies, so this budget must in turn be
distributed among these copies.  This is why the argument passed down to function $N$ is $\frac\eta{K_\eta+1}$.

Applying this unfolding for every location gives us an acyclic \wtg of depth
$\Nind{\epsilon_n,n,n}$, whose value differ from $\Value{\Gcal'}$ by at most $\epsilon$.
\end{proof}

\subsection{Zooming out}

 $\tilde{\Gcal}$ can be seen as a directed acyclic graph where each node is an SCC.
Since $\Gcal$ has a "time bound", the event ``a non-null upper clock is reset'' can happen at most $d\cdot \Tmax + (d-1)$ times on any play.
Thus the event ``a non-null upper clock is reset'' can never be part of a cycle in  $\tilde{\Gcal}$ (see \cref{rem-trim}).
Thus every SCC in $\tilde{\Gcal}$ is "up-driven".

 Any path of this graph is of length at most $n$, and each SCC contains at most $n$ locations.
 By applying \cref{VIA-convergence-rate} to every  node of the directed acyclic graph starting from the sinks, we obtain a convergence speed of $\Nind{\epsilon_n, n, n}$.

Since the output functions passed to every node are over-approximations of the value function produced by the "value iteration algorithm", any node satisfies the ``same $\slopemax$ as $\Gcal$'' condition. This suffices to prove \cref{general-rate}.

%\subsection*{Adding transition weights to cycle}

%\begin{theorem}
%\label{general-rate}
%Let $\Gcal$ be a \wtg with a "time bound" $\Tmax$ and $d$ clocks. Let $n$ be the number of locations of $\RC{\Gcal}$. 
%For any $0<\epsilon \leq \frac{n^2}{2}$ let $\epsilon_n=\frac{\epsilon}{ n^2}$.
%Now let
%$N'_\epsilon= \Nind{\epsilon_n, n, n}\cdot n+\MM$.
%%T \cdot \left( \frac{2T^2\wmax d}{\epsilon}\right)^d$.
% Then $\opt{}{N'_\epsilon}$ is an $\epsilon$-approximation of the value function of $\Gcal$.
% 
%\end{theorem}

%\begin{remark}[Dealing with transition weights]
%\todo{correct it}

%For the sake of simplicity, we assumed that every cycle has zero transition weight.
%This assumption can be removed, provided that $\epsilon$ is sufficiently small, without
%affecting the convergence bound of \cref{VIA-convergence-rate}.

\subsection{Complexity}

\begin{restatable}[Bounding the convergence rate]{remark}{unfoldingN}
\label{unfolding-N}
For $d\geq 2$, if we over-approximate the inductive formula of $N$, we obtain:
\begin{align*}
N'_\epsilon \leq \left(\frac{2d\cdot\slopemax+d +1
}{\epsilon}\right)^{2^{ n}-1}
\cdot n^{2^{n+1}}\,.
\end{align*}

\end{restatable}
See proof in appendix.

\begin{lemma}
\label{lem-deciding-time-bounded}
Deciding if  a \wtg is time-bounded can be done in exponential time.
\end{lemma}
\begin{proof}
A region trimmed $[0,1)$-\wtg is time-bounded if and only if every SCC is "up-driven". Indeed, at most one time unit can be spent in an "up-driven" SCC. Conversely, if the event ``a non-null upper clock is reset'' happens along a cycle, then this cycle can yield a run with arbitrarily long elapsed time: 

Transforming a \wtg into a region trimmed $[0,1)$-\wtg and checking every SCC can be done in polynomial time in the size of the input when guards are given in unary, thus in exponential time when given in binary.
\end{proof}

\begin{tcolorbox}
\paragraph*{Approximation Algorithm}
\begin{description}
\item[Input:] A time-bounded \wtg $\Gcal=(L_\mathsf{Min}, L_\mathsf{Max},G,\mathcal{X},T,w)$, with $m$ locations and $d$ clocks, and an initial configuration $(\ell_0,\nu_0)$, and a precision $\epsilon>0$. Weights and guards are given in binary.
\vspace{1em}
\item[Ensures:] Return $V$ such that $|V-\Value\Gcal((\ell_0,\nu_0))|\leq \epsilon$ if $V\in\Rbb$ and $V=\Value\Gcal((\ell_0,\nu_0))$ otherwise.
\begin{description}
\vspace{1em}
\item[Step $1$:] Transform $\Gcal$ into a "trimmed region $\intfo01$-\wtg" $\Gcal'$.
\vspace{1em}
\item[Step $2$:] Check that $\Gcal'$ is time-bounded (every SCC is up-driven). Let $\Tmax$ be the longest path in the directed acyclic graph underlying $\Gcal'$.
According to \cref{lem-deciding-time-bounded}, $\Tmax$ is an upper bound on the time bound of $\Gcal$.
\vspace{1em}
\item[Step $3$:] Let $\Mcal$ be the "min-cost reachability game" underlying $\Gcal'$. Check if $\Value\Mcal(\ell_0)\in\{-\infty,+\infty\}$. If yes, \textbf{return} 
$\Value\Mcal(\ell_0)$.
\vspace{1em}
\item[Step $4$:] Otherwise, use the "value iteration algorithm" on $\Gcal$ to compute $\opt{}{K}$, where
\begin{align*}
& \wmax= \max\setof{|x-y|}{x,y\in w(L)\cup\{0\}}
&&  m'=m\cdot d! \cdot (2\Tmax+1)^{d}
\\[10pt]
& 
n= m'\cdot (4(m'W+\wmax \Tmax)+1)
&& W=\max\setof{|w(t)|}{t\in T}
\\[7pt]
&K=\left(\frac{2d\cdot\slopemax+d +1
}{\epsilon}\right)^{2^{ n}-1}
\cdot n^{2^{n+1}}
&&  \slopemax=2\Tmax\wmax d
\end{align*}
\textbf{Return}  $\opt{\ell_0}{K}(\nu_0)$.
\end{description}
\end{description}
\end{tcolorbox}

Steps $1$ to $3$ are exponential in size of the input.
For Step $4$, the convergence rate $k$ is triple-exponential in size of the input.
According to \cite{alur2004},  each step $i$ of the "value iteration algorithm" is exponential in $\RC{\Gcal}$ and $i$. Thus Step $4$ is quadruple-exponential in size of the input.
Therefore
 the "Approximation problem" for "time-bounded" \wtg is in $4$-\textsc{ExpTime}. 
 
%As a comparison point, the approximation algorithm given in \cite{Busatto2018} for almost-divergent weighted timed games is doubly-exponential in the size of $\RC{\Gcal}$.

Note that, for simplicity of presentation, the bound given on the convergence speed is far from tight. Indeed, when applying the induction step in \cref{lem-induction}, we assume that, when unfolding $\Gcal'$ into a string of $K_\eta+1$ copies,
one time unit can then be spent in each copy, whereas this time bound actually concern the set of copies.
However, taking time-sharing into account still gives us a triple-exponential convergence rate.

\section*{Conclusion}

We have shown that the "value iteration algorithm" converges toward the value for time-bounded \wtgs, with a convergence speed bounded in $3$\textsc{ExpTime}, thus proving decidability of the "Approximation problem" on turn-based time-bounded weighted timed games in $4$\textsc{ExpTime}

The main open problem that remains is the decidability of Approximation for \wtgs with non-negative weights only. We hope that the techniques developed in this paper may provide useful insights toward resolving this question. Indeed, in a \wtg without negative weights, the time spent in locations with positive weight is effectively bounded by any greedy upper bound on the value.

\clearpage

% references
%\bibliographystyle{alpha}
% for natbib
%\bibliographystyle{plainurl}
\bibliography{biblio} 

\clearpage

\appendix
\section{Proofs of \cref{sec-bisim}}
%\todo{remettre W1 W2}
%\begin{theorem}
%\label{thm-Lip-TB}
%Let $\Gcal$ be a "trimmed region \wtg" with a time bound $T$. Let $\wmax$ be the maximal difference between two elements of $\{0\}\cup\{w(\ell):\ell\in L\}$, and let $d$ be the number of clocks. Then for any location $\ell$, the value function $\nu\mapsto \Value\Gcal(\ell,\nu)$ is is Lipschitz-continuous, of slope at most $=2T\wmax d $.
%\end{theorem}

To prove \cref{thm-Lip-TB}, we use the same approach as in \cite{vialard2025decidingvaluetwoclocknonzeno}:

\begin{definition}[Simulation and bisimulation]

Let $\mathcal{G}_1=(L_\Minsf,L_\Maxsf,G,\mathcal{X},T,w)$ and $\mathcal{G}_2=(L'_\Minsf,L'_\Maxsf,G',\mathcal{X},T',w')$ be \wtgs. Then a relation $R\subseteq (L_\Minsf\times \Rbb^\Xcal\times L'_\Minsf\times \Rbb^\Xcal) \cup (L_\Maxsf\times \Rbb^\Xcal\times L'_\Maxsf\times \Rbb^\Xcal)$ is a \textbf{simulation relation} when
\[(\ell_1,\nu_1)\,R\,(\ell_2,\nu_2) \implies \begin{cases}
\medskip
\text{$\ell_1\in G$ and $\ell_2\in G'$,}\\
\text{or, for all $t_1\in T$, $\delta_1\in\Rbb$, there exist $t_2\in T'$ and $\delta_2\in\Rbb$ such that if }\\
 \text{$(\ell_1,\nu_1)\xrightarrow{\delta_1,t_1}(\ell'_1,\nu_1')$ then $(\ell_2,\nu_2)\xrightarrow{\delta_2,t_2}(\ell'_2,\nu'_2)$ where $(\ell_1',\nu_1')\,R\,(\ell'_2,\nu'_2)$.}
\end{cases}\]

$R$ is a ""\textbf{bisimulation}"" if $R$ and its converse are both simulations.
\end{definition}

\begin{example}
 The region relation $R$ between configurations of the same \wtg, where $(l_1,\nu_1)\,R\,(l_2,\nu_2)$ iff $l_1=l_2$ and $\nu_1,\nu_2$ are in the same region, is a "bisimulation".
\end{example}

A "(bi)simulation" relation $R$ can be extended to runs:
$\rho_1 \,R\, \rho_2$ when $|\rho_1|=|\rho_2|$, and  $\rho_1(n) \,R\, \rho_2(n)$ for all $n<|\rho_i|$.

\begin{lemma}[Lemma 23 from \cite{vialard2025decidingvaluetwoclocknonzeno}]
\label{lem-bisim-WTG}
Let $\Gcal_1$ and $\Gcal_2$ be two \wtgs. 
%Let $V_{\Gcal}$ and 
%$V_{\Gcal'}$ be the sets of cost functions on transitions to $G$.
Let $\,R\,$ be a "bisimulation" relation between $\Gcal_1$ and $\Gcal_2$, such that $\Gcal_1$ and $\Gcal_2$ start from configurations $c_1,c_2$, respectively, with $c_1 \,R\, c_2$. Assume that $\Value{\Gcal_i}\neq +\infty$ for $i\in \{1,2\}$. 
Then 

\begin{equation*}\label{eq-bisim}
|\Value{\Gcal_1}(c_1)-\Value{\Gcal_2}(c_2)|\leq \sup
\enstq{|\weight(\rho_1)-\weight(\rho_2)|}{\begin{array}{c}
		\rho_1,\rho_2 \text{ plays of } \Gcal_1,\Gcal_2\\ 
		\text{starting in $c_1,c_2$,}\\
		\rho_1 \,R\, \rho_2\\
		\end{array}}
.
\end{equation*}
\end{lemma}

%
%Here we need a stronger version of \cref{lem-bisim-WTG}. Let $P$ be some property on pairs of delayed transitions. 
%We say that $R$ is a bisimulation respecting $P$ if
%\[(\ell_1,\nu_1)\,R\,(\ell_2,\nu_2) \implies \begin{cases}
%\medskip
%\text{$\ell_1\in G$ and $\ell_2\in G'$,}\\
%\text{or, for all $t_1\in T$, $\delta\in\Rbb$, there exist $t_2\in T'$ and $\delta_2$ such that if }\\
% \text{$(\ell_1,\nu_1)\xrightarrow{\delta_1,t_1}(l_1',\nu_1')$ then $(\ell_2,\nu_2)\xrightarrow{\delta_2,t_2}(\ell'_2,\nu'_2)$ where $(\ell_1',\nu_1')\,R\,(\ell'_2,\nu'_2)$}\\
%  \text{and $(\ell_1,\nu_1)\xrightarrow{\delta_1,t_1}(l_1',\nu_1')$, $(\ell_2,\nu_2)\xrightarrow{\delta_2,t_2}(\ell'_2,\nu'_2)$ respect $P$.}\\
%\end{cases}\]

Here we need a stronger version of \cref{lem-bisim-WTG}. Let $P$ be some property on pairs of plays. 
We say that $R$ is a bisimulation respecting $P$ if
for any pair of plays $\rho_1 R \rho_2$ that verifies $P$, where $c_0\xrightarrow{\rho_1} (\ell_1,\nu_1)$, $c'_0\xrightarrow{\rho_2} (\ell_2,\nu_2)$, we have
$$\begin{cases}
\medskip
\text{$\ell_1\in G$ and $\ell_2\in G'$,}\\
\text{or, for all $t_1\in T$, $\delta\in\Rbb$, there exist $t_2\in T'$ and $\delta_2$ such that if }\\
 \text{$(\ell_1,\nu_1)\xrightarrow{\delta_1,t_1}(l_1',\nu_1')$ then $(\ell_2,\nu_2)\xrightarrow{\delta_2,t_2}(\ell'_2,\nu'_2)$ where $(\ell_1',\nu_1')\,R\,(\ell'_2,\nu'_2)$}\\
  \text{and $(\rho_1(\delta_1,t_1), \rho_2(\delta_2,t_2))$ verifies $P$.}\\
\end{cases}$$

When every pair of matching length prefixes of two runs $\rho_1 R \rho_2$ verifies $P$, we write $\rho_1 R^P \rho_2$.
Then \cref{lem-bisim-WTG} can be enhanced into the following result.

\begin{lemma}
\label{lem-bisim-WTG-more}
Let $\Gcal_1$ and $\Gcal_2$ be two \wtgs. 
Let $R$ be a bisimulation relation between $\Gcal_1$ and $\Gcal_2$ respecting some property $P$, such that $\Gcal_1$ and $\Gcal_2$ start from configurations $c_1 \,R\, c_2$.

 Assume that $\Value{\Gcal_i}\neq \pm\infty$ for $i\in \{1,2\}$.
%  and let $W_i>\Value{\Gcal_i}$.
Then 

\begin{equation*}\label{eq-bisim}
|\Value{\Gcal_1}(c_1)-\Value{\Gcal_2}(c_2)|\leq \sup
\enstq{|\weight(\rho_1)-\weight(\rho_2)|}{\begin{array}{c}
		\rho_1,\rho_2 \text{ finite plays of } \Gcal_1,\Gcal_2\\ 
		\text{starting in $c_1,c_2$,}\\
		\rho_1 \,R^P\, \rho_2\\
%		\weight(\rho_1)\leq W_1, \weight(\rho_2)\leq W_2
		\end{array}}
\;.
\end{equation*}

% $$|V(\Gcal_1)-V(\Gcal_2)|\leq \max\setof{|w_1(\ell_1,\nu_1)-w(\ell_2,\nu_2)|}{\ell_1\in G_1, \ell_2 \in G_2, (\ell_1,\nu_1) \,R\, (\ell_2,\nu_2)}\;.$$
\end{lemma}
The proof is almost identical to the proof of \cref{lem-bisim-WTG}: 
\begin{proof}
Consider $\sigma^1_\Mimin$ a strategy for $\Minsf$ on $\Gcal_1$  and $\sigma^2_\Mamax$ a strategy for Max on $\Gcal_2$. Here, let us choose 
$\sigma^1_\Mimin$ such that 
$\sup_{\sigma_\Mamax} V_{\Gcal_1}(\sigma^1_\Mimin, \sigma_\Mamax)< +\infty$. Observe that 
$$\Value{\Gcal_1}=\inf \setof{\sup_{\sigma_\Mamax} \weight(\play_{\Gcal_1}(c_1,\sigma_\Mimin, \sigma_\Mamax))}{ \sigma_\Mimin \text{ s.t. } \sup_{\sigma_\Mamax} V_{\Gcal_1}(\sigma_\Mimin, \sigma_\Mamax)< +\infty}.$$
In other words, we can consider only ``smart'' strategies for $\Minsf$ when defining the value.  Similarly, let us choose 
$\sigma^2_\Mamax$ such that $\inf_{\sigma_\Mimin} V_{\Gcal_2}(\sigma_\Mimin, \sigma^2_\Mamax)< +\infty$
%This is the main difference we made to the proof of \cref{lem-bisim-WTG}.

By definition of bisimulation, we can build
 $\sigma^2_\Mimin$ and $\sigma^1_\Mamax$ strategies for Min and Max on $\Gcal_2$ and $\Gcal_1$ respectively, such that $\rho_1 \,R^P\, \rho_2$ for $\rho_i=\play(c_i,\sigma^i_\Mimin,\sigma^i_\Mamax)$ with $i\in\{1,2\}$.
(In any Min location, $\sigma^2_\Mimin$ simulates $\sigma^1_\Mimin$ following $R$ and $P$. Similarly, in any Max location, $\sigma^1_\Mamax$ simulates $\sigma^2_\Mamax$.)
Because $\weight(\rho_1)<+\infty$, we know that $\rho_1$, and thus $\rho_2$, are finite plays.

Since the set of strategies obtained by such a simulation is included in the set of all strategies,
\begin{align*}
\sup_{\sigma^2_\Mamax} \weight(\rho_1) \leq \sup_{\sigma_\Mamax} V_{\Gcal_1}(\sigma^1_\Mimin, \sigma_\Mamax)\text{ for a fixed $\sigma^1_\Minsf$,}\\
\shortintertext{and}
\inf_{\sigma^1_\Mimin} \weight(\rho_2) \geq \inf_{\sigma_\Mimin} V_{\Gcal_2}(\sigma_\Mimin, \sigma^2_\Mamax)\text{ for a fixed $\sigma^2_\Maxsf$.}
\end{align*}

Therefore
\begin{align*}
 \inf_{\sigma^1_\Mimin} 
\sup_{\sigma^2_\Mamax}
\weight(\rho_1)&\leq 
\inf_{\sigma^1_\Mimin}\sup_{\sigma_\Mamax} \weight(\play_{\Gcal_1}(c_1,\sigma_\Mimin, \sigma_\Mamax))= \Value{\Gcal_1}\\
\shortintertext{and}
\inf_{\sigma^1_\Mimin}
\sup_{\sigma^2_\Mamax} 
\weight(\rho_2)
&\geq
\sup_{\sigma^2_\Mamax} 
\inf_{\sigma^1_\Mimin}
\weight(\rho_2)\\
&\geq 
\sup_{\sigma_\Mamax}\inf_{\sigma_\Mimin} \weight(\play_{\Gcal_2}(c_2,\sigma_\Mimin, \sigma_\Mamax))= \Value{\Gcal_2}
\end{align*}

Therefore $\Value{\Gcal_2}-\Value{\Gcal_1}\leq 
\sup\setof{\weight(\rho_2)-\weight(\rho_1)}{\rho_1 \,R^P\, \rho_2}\;.$

%
%\sup_{\sigma^1_\Mimin \sigma^2_\Mamax} V(\rho_2(\sigma^1_\Mimin,\sigma^2_\Mamax))-V(\rho_1(\sigma^1_\Mimin,\sigma^2_\Mamax))$

Reasoning in mirror, one can obtain
\[\Value{\Gcal_1}-\Value{\Gcal_2}\leq 
\sup\setof{\weight(\rho_1)-\weight(\rho_2)}{\rho_1 \,R^P\, \rho_2}\;.\]

%
%$\sup_{\sigma^2_\Mimin \sigma^1_\Mamax} V(\rho_1(\sigma^2_\Mimin,\sigma^1_\Mamax))-V(\rho_2(\sigma^2_\Mimin,\sigma^1_\Mamax))$ 
and combine to conclude.

\end{proof}

Now let us define the following relation between valuations:
For $\epsilon>0$, two valuations $\nu,\nu'$ are \textbf{$\epsilon$-neighbours} if there exists $\epsilon_1,\epsilon_2\geq 0$ such that $\epsilon_1+\epsilon_2=\epsilon$, and for any clock $x\in\Xcal$, $\nu(x)-\nu'(x)\in[-\epsilon_1,\epsilon_2]$.

\begin{lemma}
\label{lem-bisim-spicy-all}

For any configurations $(\ell,\nu)$ 
and $(\ell',\nu')$ of a "trimmed region \wtg" $\Gcal$.

Let $(\ell,\nu) \,R_\epsilon\, (\ell',\nu')$ \tiff 
 $\ell=\ell'$,
and  $\nu$ and $\nu'$ are $\epsilon$-neighbours

Then $R_\epsilon$ is a bisimulation relation respecting the following property $P$: 
%\textcolor{red}{delays are equal, except when one clock is close to $1$}

We say that a pair $(\rho_1,\rho_2)$ respect $P$ iff for any pair of matching delayed transitions 
 $((\ell,\nu_1)\xrightarrow{\delta_1,t_1} c_1,(\ell,\nu_2)\xrightarrow{\delta_2,t_2} c_2)$ along $(\rho_1,\rho_2)$,
 we have $t_1=t_2$
 and for some upper clock $x$ of $\reg{\ell}$, 
 $|(\nu_1+\delta_1)(x) - (\nu_2 + \delta_2)(x)|= |\nu_1(x) -\nu_2(x) | - |\delta_1 - \delta_2|$. In other words, the choice of $\delta_1$ and $\delta_2$ brings the value of upper clocks of the two plays closer together.
\end{lemma}
\begin{proof}

Let $t:\ell_1\xrightarrow{C,X} \ell_2$. Let $\nu_1,\nu_1'\in \reg{\ell_1}$ be $\epsilon$-neighbours. By definition, there are $\epsilon_1,\epsilon_2\geq 0$ such that $\epsilon_1+\epsilon_2=\epsilon$ and for any clock $x\in\Xcal$, $\nu_1(x)-\nu_1'(x)\in[-\epsilon_1,\epsilon_2]$.

Let $\delta\geq 0$ such that $\nu_1+\delta \models C$ and let $\nu_2=(\nu_1+\delta)[X:=0]$. Then let us show that there exists $\delta'\geq 0$ such that $\nu'_1 + \delta' \models C$,
and
$(\ell_2,\nu_2)R_\epsilon (\ell'_2,\nu'_2)$ with $\nu'_2=(\nu'_1+\delta')[X:=0]$. To prove that, one only needs to show  that $\nu_1+\delta$ and $\nu'_1+ \delta'$ are $\epsilon$-neighbours.

%Consider the interval
%\[\Delta=\enstq{\delta'\geq0}{\begin{array}{c}
%		\nu'_1+\delta'\models C\\ 
%		\nu'_1+\delta'\in r\text{ s.t. }\nu_1'+\delta' \in \adh{r}\\
%\end{array}}\]
Since $\Gcal$ is a "region trimmed \wtg", there exists $\delta'$ such that $\nu'_1+\delta'\models C$
%, and every valuation of the form $\nu+\delta$ that satisfy $C$ with $\nu\in \reg{\ell_1}$ belong to a unique region $r_t$. Since $\nu'_1\in \regadh{\ell_1}$ and $\nu'_1 + \delta' \models \overline{C}$, $\nu'_1 + \delta' \in \adh{r_t}$.
%Therefore, $\Delta$ is not empty.

If $\nu'_1+\delta\models C$, then $\delta'=\delta$ preserves $\epsilon$-neighbourhood.

Otherwise, some guard in $C$ is at play. 
This guard cannot be of the form $x=0$ for some clock $x$, otherwise  a delay $0$ would work for both $\nu_1$ and $\nu'_1$. Moreover, it cannot be of the form $x>0$ since $\nu_1$ and $\nu'_1$ are in the same region, thus we would have $\delta>0$. Therefore, the guard which forbids delay $\delta$ from $\nu'_1$ is of one of these forms: $x=1$,  $x<1$ or $x\leq 1$, where $x$ is an "upper clock" of $\reg{\ell_1}$.

		\begin{itemize}
			\item If there is a guard $x=1$ in $C$ for some $x\in\Xcal$, then $\delta=1-\nu_1(x)$, and we can only pick $\delta'=1-\nu'_1(x)$.
 Moreover, for all $y\in\Xcal$,  we know that
			$$
			(\nu_1(y) + \delta)-(\nu_1'(y) + \delta') \in\intff{-\epsilon'_1}{\epsilon'_2}
			,$$
			for $\epsilon_2'\eqdef \epsilon_2-(\delta'-\delta)\geq 0$, and $\epsilon_1'\eqdef \epsilon_1+(\delta'-\delta)\geq 0$.
			Since $\delta'-\delta=\nu_1(x)-\nu_1'(x)\in\intff{-\epsilon_1}{\epsilon_2}$, $\epsilon'_1\geq 0$ and $\epsilon'_2\geq 0$,
		 and $\epsilon_1'+\epsilon_2'=\epsilon_1+\epsilon_2=\epsilon$.
\item  Otherwise, there must be a guard $x< 1$  or $x\leq 1$ in $C$ for some $x\in\Xcal$, 
%such that $\nu_1+\delta\models (x\leq 1)$ but $\nu'_1+\delta' \not
%\models (x< 1)$ (resp. $x\leq 1$), then $\Delta<\delta'$.
If $\delta + \nu_1(x) - \nu'_1(x)> 0$ then pick $\delta'=\delta + \nu_1(x) - \nu'_1(x)$.
			Since $\delta'-\delta=\nu_1(x)-\nu_1'(x)\in\intff{-\epsilon_1}{\epsilon_2}$, $\nu_1+\delta$ and $\nu_1'+\delta'$ are $\epsilon$-neighbours.

Otherwise,
 $\delta \leq \nu'_1(x)-\nu_1(x) \leq \epsilon_1$
 Then take any $0<\delta'$ such that $\nu'_1(x)+\delta'<1$.

 Let
$\epsilon_2'\eqdef \epsilon_2+\delta-\delta'\geq 0$ and $\epsilon_1'\eqdef \epsilon_1-(\delta-\delta')\geq0$.

Observe that $\delta'<\delta$, thus $\epsilon'_2\geq0$. Moreover $\epsilon_1\geq 0$ since $\delta-\delta'<\delta \leq \epsilon_1$.
Finally, $\epsilon_1'+\epsilon_2'=\epsilon$, hence $\nu_1' + \delta'$ and $\nu_1+\delta$ are $\epsilon$-neighbours.\qedhere
			
\end{itemize}

\end{proof}

\LipTB*

\begin{proof}[Proof of \cref{thm-Lip-TB}]
For any $0<\epsilon$, for any valuations $\nu,\nu'$ such that $\distinf{\nu-\nu'}\leq \epsilon$, we know that $\nu,\nu'$ are $\epsilon'$-neighbours for $\epsilon'=2\epsilon$.
According  to \cref{lem-bisim-WTG-more},

\begin{equation*}
|\Value{\Gcal_1}(c_1)-\Value{\Gcal_2}(c_2)|\leq \sup
\enstq{|\weight(\rho_1)-\weight(\rho_2)|}{\begin{array}{c}
		\rho_1,\rho_2 \text{ finite plays of } \Gcal_1,\Gcal_2\\ 
		\rho_1 \,R_{\epsilon'}^P\, \rho_2
		\end{array}}
\;.
\end{equation*}
For any pair of plays $\rho_1 \,R_{\epsilon'}^P\, \rho_2 $, we know that
$|\weight(\rho_1)-\weight(\rho_2)|\leq 2\epsilon \cdot \Tmax\wmax d$. Indeed, delay differences bring the upper clocks together by definition of $P$. Thus, if we look at a section of $\rho_1$ and $\rho_2$ side by side with some upper clock $t$, the delay difference is at most $\epsilon'$ in this section.  Accumulating a delay difference of $\epsilon'$ implies that the clock $t$ has the same valuation in both plays. Thus there cannot be more delay differences until $t$ and every upper clock has been reset. 

Here $\Gcal$ has a time bound $\Tmax$; 
%according to \cref{lem-deciding-time-bounded}, every SCC of $\Gcal$ is "up-driven".
 this means that if we add a clock $t$ which is automatically reset whenever it reaches $1$, 
then any clock $x$ can become a non-null upper clock at most $\Tmax$ times in a play before $t$ reaches $1$ for the $\Tmax$-th time.
Thus a delay difference of at most $\epsilon'$ can be accumulated 
at most $d\Tmax$ times, and thus induce a value difference of at most $\Tmax \wmax d\epsilon' = 2\Tmax \wmax d\epsilon$.\qedhere
%\todo{should I add t earlier ? what about the decision procedure? or is it not needed?}

%the delays taken in $\rho_1,\rho_2$ only differs when a clock is $2\epsilon$ close to $1$. This happens at most $T$ times for every clock, and the weight difference caused is at most $2\epsilon \cdot\wmax$.

\end{proof}

\correlax*
\begin{proof}[Proof of \cref{cor-relax}]
Let $\Gcal$ be a "trimmed region \wtg" and $\Gcal'$  a relaxed copy of $\Gcal$. Then for any $\epsilon$, the relation $(\ell,\nu)\,R'_\epsilon\,(\ell,\nu')$ defined as
\begin{itemize}
\item $\nu,\nu'$ are $\epsilon'$-neighbours for some $\epsilon'<\epsilon$,
\item $\nu\in\reg{\ell}$ and $\nu'$ belongs to the closure of $\reg{\ell}$,
\end{itemize}
is a bisimulation relation respecting property $P'$, defined as follows:
A pair of plays $(\rho_1,\rho_2)$
iff for any pair of matching delayed transitions $((\ell,\nu_1)\xrightarrow{\delta_1,t_1}(\ell',\nu_1'),(\ell,\nu_2)\xrightarrow{\delta_2,t_2}(\ell',\nu'_2))$ along $(\rho_1,\rho_2)$, we have:
\begin{itemize}
\item $t_1=t_2$
\item for some upper clock $x$ of $\reg{\ell}$, either 
 $|(\nu_1+\delta_1)(x) - (\nu_2 + \delta_2)(x)|= |\nu_1(x) -\nu_2(x) | - |\delta_1 - \delta_2|$, or $\nu_1(x)=\nu_2(x)=0$.
 \item If $\nu_1(x)=\nu_2(x)=0$ and $|\delta_1 - \delta_2|>0$, then $|\delta_1 - \delta_2|\leq \frac\delta 2$, where $\delta$ is the delay difference that occurred the previous time such a pattern was seen along $\rho_1,\rho_2$ ($+\infty$ if undefined). 
\end{itemize}

In other words, delay difference must bring the upper clocks closer if possible.
The only case where this is not possible is when transitions with constraints $x>0$ and $x\geq 0$ leave a singleton-region location. Then $P'$ enforces that the sum of these delay differences converges.

% $(\ell,\nu_1)\xrightarrow{\delta_1,t_1}(\ell',\nu_1')$, $(\ell,\nu_2)\xrightarrow{\delta_2,t_2}(\ell',\nu'_2)$ respects $P'$
% iff $t_1=t_2$
% and for some upper clock $x$ of $\reg{\ell}$, either 
% $|(\nu_1+\delta_1)(x) - (\nu_2 + \delta_2)(x)|= |\nu_1(x) -\nu_2(x) | - |\delta_1 - \delta_2|$, or $\nu_1(x)=\nu_2(x)=0$ (this last possibility accounts for possible $(x>0,x\geq 0)$ pair of constraints where bringing upper clocks closer is not possible).

Note that, since $\Gcal$ is "trimmed" there are no transitions that can be taken from $(\ell,\nu')$ in $\Gcal'$ and not from $(\ell,\nu)$ in $\Gcal$. 
 Thus the proof is similar to \cref{lem-bisim-spicy-all}, with one notable adjustment:
we ask for $\epsilon'<\epsilon$-neighbourhood because the guard differences between $\Gcal$ and $\Gcal'$ might force an arbitrarily small increase in the distance between $\nu$ and $\nu'$.

Now, for any $\rho_1 R^{P'} \rho_2$ we have
$|\weight(\rho_1)-\weight(\rho_2)|\leq \epsilon \cdot \Tmax\wmax d +2\epsilon \wmax$. The first part of the sum corresponds to delay differences bringing clock together, the second part to the sum of delay differences from singleton-region location, which converges.

Therefore, according to \cref{lem-bisim-WTG-more}, the value function of $\Gcal'$ corresponds to the value function of $\Gcal$ extended continuously to the closure of regions. \qedhere

%Now let $C$ be a set of constraints, and $\overline{C}$ be a copy of $C$ where every guard has been relaxed. Then for any valuation $\nu$ that verifies 
%$\overline{C}$ but not $C$, and for any $\epsilon>0$, there exists $\nu_\epsilon\models C$ a valuation $\epsilon$-close to $\nu$.
\end{proof}

\section{Proofs of \cref{sec-approx}}

\subsection{Proof of \cref{lem-no-neg-cycle}}

\begin{proof}
Let $\Mcal$ be the "min-cost reachability game" obtained by removing clocks, guards and location weights from $\Gcal$. Observe that
$|\Value\Mcal-\Value\Gcal| \leq \Tmax\wmax$, since any strategy can be lifted from one model to the other (see \cref{rem-trim}), and since weight due to delay is contained within $-\Tmax\wmax$ and $\Tmax\wmax$.

In particular, this implies that the value of $\Mcal$ is also finite. Furthermore, Brihaye \etal\cite{BGHM16} show that
 $|\Value\Mcal|\leq nW$~\footnote{\cite{BGHM16}, Proof of Proposition $4$.}.
 Intuitively, if a path takes no cycles, it reaches the goal location in at most $n$ step. Since the value is finite, $\Minsf$ (resp. $\Maxsf$) is able to evade any cycle of transition weight $\geq 0$ (resp. $<0$), thus able to guarantee a cost $\leq nW$ (resp. $\geq -nW$).
 
 Consider a strategy $\sigma_\Minsf$ that yields a run $\rho$ in $\Gcal$ such that the accumulated transition weight along $\rho$  is $B+1$ or above, from some configuration $c_0$ to a configuration $c$. 
We know that 
$\sup_{\sigma_\Maxsf} \weight(\play_\Gcal(\sigma_\Minsf,\sigma_\Maxsf))\geq 
\weight(\rho\rho')$ such that $\rho\rho'=\play_\Gcal(\sigma_\Minsf,\sigma_\Maxsf)$ for some $\sigma_\Maxsf$. 

Furthermore, since $|\Value\Mcal|\leq nW$, we can pick $\sigma_\Maxsf$ such that $\rho'$ has an accumulated transition weight of at least $-nW$. 
Thus $\weight(\rho\rho')\geq -nW -\Tmax\wmax +B +1 \geq nW +\Tmax\wmax +1\geq \Value\Gcal(c_0) +1$. Thus we can ignore any such $\Minsf$ strategy when computing the Value.

Symmetrically, we can ignore any $\Maxsf$ strategy that yields a run of accumulated transition weight $\leq -B-1$.
 
Therefore the accumulated cost along relevant runs is bounded, which means that it can be stored inside locations.
Therefore we can construct an equivalent \wtg	$\Gcal'=(L'_\mathsf{Min}, L'_\mathsf{Max},G,\mathcal{X},T',w')$
 such that:
 \begin{itemize}
 \item $L'=L\times \set{-B, -B+1,\dots, B} \cup \{\ell_{-\infty},\ell_{+\infty}\}$ with $B=2\Tmax\wmax+2nW$.
 \item $\ell_{\pm\infty}$ are locations of value $\pm \infty$.
 \item for any $t:\ell\xrightarrow{C,X} \ell' \in T$, for any $i\in\set{-B,\dots, B}$, if $i+w(t)\in  \set{-B,\dots, B}$, add $t:(\ell,i)\xrightarrow{C,X} (\ell',i+w(t))$ to $T'$. Otherwise add 
 $t:(\ell,i)\xrightarrow{C,X} (\ell_{-\infty})$ or $t:(\ell,i)\xrightarrow{C,X} (\ell_{+\infty})$ to $T'$, depending on the sign of $i+w(t)$.
 \item Let $w'((\ell,i))=w(\ell)$. For any $t\in T'$ reaching a location $(\ell,i)$, let $w'(t)=i$ if $\ell$ is a goal location, or $0$ otherwise.
 \end{itemize}
 Remove from $\Gcal'$ any infinite-value locations caused by $\ell_{-\infty},\ell_{+\infty}$.

This construction ensures that for any location $\ell$, we now have $\Value\Gcal(\ell,\cdot)=\Value{\Gcal'}((\ell,0),\cdot)$.
\end{proof}

%\VIArateSCC*
%\leastweight*
%\SCCinduction*
%\todo{mute VIArateSCC et SCCinduction}
\subsection{Proof of \cref{VIA-convergence-rate}}

\AP A location $\ell \in L$ is said to be ""\emph{urgent}"" if every outgoing transition from $\ell$ either leaves the \wtg or has a guard of the form $y=0$ for some clock $y$.

The following lemma is the main building block of the proof of \cref{VIA-convergence-rate}. Here we show that if $\Gcal'$ has a "non-urgent" location of minimal weight $\ell$, then it can be unfolded into a string of $K_\eta+$ copies of itself, for some $\eta>0$, which count (and thus bound by $K_\eta$) how many times a positive delay has been spent in $\ell$. Each copy has one less "non-urgent" location. The value difference created by this transformation is at most $\eta$.

\begin{restatable}[How to make a location of least weight urgent]{lemma}{leastweight}
\label{dealing-with-least-weight-state}

 Let $L$ be the set of locations of ${\Gcal'}$
and let
$
0 < n \leq |L|
$
be the number of "non-urgent" locations in ${\Gcal'}$.
Let
$
\ell \in L
$
be a "non-urgent" location of minimal weight in ${\Gcal'}$.

Then, for any $\eta>0$, there exists a WTG $\Gcal''$ made of a string of
$$
K_\eta +1 = d \cdot \left(\left\lfloor \frac{2\slopemax}{\eta}\right\rfloor +1\right) +1
$$
copies of ${\Gcal'}$, such that $\distinf{\Value{\Gcal'}(\ell'_0,\cdot)-\Value{\Gcal''}(\ell''_0,\cdot)}\leq \eta$ ($\ell''_0$ is the first copy of the initial location $\ell'_0$ of $\Gcal'$ in $\Gcal''$).

In $\Gcal''$, outgoing transitions from $\ell$ are duplicated such that waiting in $\ell$ implies jumping from one copy to the following one. 
In the last copy, waiting in $\ell$ is not possible.
 \end{restatable}

\begin{proof}[Proof of \cref{dealing-with-least-weight-state}]

Fix
$
\eta'=\frac{\eta}{2\slopemax}
$. Let $x$ be an upper clock of $\Gcal'$. 

\begin{description}
%%%%%%%%%%%%%%%%%%%%%%%%%%%%%%%%%%%%%%%%%%%%%%%%%%%%%%%%%%%%%%%%%%%%
%ELL BELONGS TO MIN
%%%%%%%%%%%%%%%%%%%%%%%%%%%%%%%%%%%%%%%%%%%%%%%%%%%%%%%%%%%%%%%%%%%%
\item[If $\ell$ belongs to $\Minsf$,]
observe that any strategy for $\Maxsf$ in $\Gcal''$ can be translated into an equivalent strategy in $\Gcal'$, and conversely.
Whereas for $\Minsf$, playing in $\Gcal''$ is equivalent  to playing in $\Gcal'$ with the additional restriction that $\Minsf$ can take at most  $K_\eta$ positive-delay transitions from $\ell$. Thus
$\Value{\Gcal''}-\Value{\Gcal'}\geq 0$. 

Let us now fix $\sigma_\Minsf$ and $\sigma_\Maxsf$ strategies for $\Minsf$ and $\Maxsf$ in $\Gcal'$. We will now build $\sigma^\eta_\Minsf$ a strategy for $\Minsf$ which will work for both $\Gcal'$ and $\Gcal''$.

\begin{itemize}

\item Initially, let $\Xflawed=\emptyset$.

\item We define $\sigma^\eta_\Minsf$ inductively by simulating $\sigma_\Minsf$. To this end, we simultaneously construct a play $\rho^\eta$ produced by the strategy profile $(\sigma^\eta_\Minsf,\sigma_\Maxsf)$ in $\Gcal'$ and a play
\footnote{this play is not actually a valid play in $\Gcal'$, since it can contain some makeshift transitions.}
 $\rho$ that verifies $\weight(\rho)\leq \sup_{\sigma'_\Maxsf}\weight
(\play(\sigma_\Minsf,\sigma'_\Maxsf))$,
 both starting from the same valuation in $\ell_0''$ and $\ell_0'$, respectively. Our construction guarantees that at each step the two plays are in matching locations. Moreover, at every step, the pair $(\rho^\eta,\rho)$ remains synchronized, in the sense that the same amount of time has elapsed in both plays since entering $\Gcal'$. Consequently, the valuation of the upper clock $x$ is identical in the corresponding configurations of the two plays at every step.
\item More formally, 
consider a pair of corresponding configurations $(\ell_\Minsf,\nu^\eta),(\ell_\Minsf,\nu)$ with $\ell_\Minsf\in L_\Minsf\setminus \{\ell\}$. Whenever $\sigma_\Minsf$ prescribes some delayed transition, let $\sigma^\eta_\Minsf$ prescribe the exact same delayed transition from  $(\ell_\Minsf,\nu)$.
Because of trimming, relaxed guards and synchronicity, this delayed transition in valid.
\item Similarly, 
consider a pair of corresponding configurations $(\ell_\Maxsf,\nu^\eta)$, $(\ell_\Maxsf,\nu)$ with $\ell_\Maxsf\in L_\Maxsf$. Whenever $\sigma_\Maxsf$ prescribes a delayed transition from $(\ell_\Maxsf,\nu^\eta)$ along $\rho^\eta$, let $\Maxsf$ follow the same delayed transition from $(\ell_\Maxsf,\nu)$ in $\rho$.
\item 
Now consider
a pair $(\ell,\nu^\eta_1)$,$(\ell,\nu_1)$, such that
 $\sigma_\Minsf$ prescribes the delayed transition
$(\ell,\nu_1) \xrightarrow{\delta,t} (\ell',\nu_2)$.
If $\delta=0$, then let $\sigma^\eta_\Minsf$ prescribe $(\delta,t)$ from $(\ell,\nu^\eta_1)$.

 Otherwise, consider the set $\Pcal$
 of plays of the form 
 $(\ell,\nu_1) \xrightarrow{\delta,t} (\ell',\nu_2)
\rightarrow^*
(\ell,\nu_3)
\xrightarrow{\delta',t'}
(\ell'',\nu_4)$
produced by $\sigma_\Minsf$ and any $\Maxsf$ strategy, such that
\begin{itemize}
\item $\delta'>0$
\item $\nu_3(x)-\nu_1(x)\leq \eta'$
\item for any clock $y$ reset along the path
$
(\ell,\nu_1)\rightarrow^*(\ell,\nu_3)
$,   $y\not\in\Xflawed$ and $\nu^\eta_1(y)\leq \eta'$.
\end{itemize}
 
If $\Pcal$ is empty, let $\sigma^\eta_\Minsf$ prescribe $(\delta,t)$.

\AP Otherwise, consider the supremum of the values $\nu_4(x)$ over all plays in $\Pcal$. There exists a transition $t'$ such that the same supremum is obtained when restricting $\Pcal$ to the subset of plays ending with $t'$. We denote this subset by $\Pcal'$.
Then let $\sigma^\eta_\Minsf$ prescribe $(\ell,\nu^\eta_1)\xrightarrow{\delta'',t'}(\ell'',\nu^\eta_4)$ such that $\nu^\eta_4(x)$ equals the supremum of all $\nu_4(x)$ in $\Pcal'$. We call this transition a ""shortcut transition"". Relaxed guards in $\Gcal'$ ensure that this delayed transition is valid.
Let $X$ be the set of clocks reset between configurations
$
(\ell,\nu_1)$ and $(\ell,\nu_3)
$ along any path in $\Pcal'$. Add $X$ to $\Xflawed$.

Let us define $\rho$ so that the "shortcut transition" in $\rho^\eta$ is matched by a play in $\Pcal'$ in $\rho$ and the pair of plays remains synchronized. 
If no play in $\Pcal'$ reaches the supremum, then $\Pcal'$ is infinite. Thus according the Bolzano–Weierstrass theorem, there exists
a sequence $(\nu'_n)$ in $\{\nu_4\in \Pcal'\}$ such that there exists $\nu'=\lim_{n\rightarrow +\infty} \nu'_n$ and
$\nu'(x)=\nu^\eta_4(x)$, and such that there exists $w'$ the limit weight of paths from $\nu_1$ to $\nu_n$ in $\Pcal'$.
Then let $\rho$ follow a makeshift transition from $(\ell,\nu_1)$ to $(\ell'',\nu')$ of weight $w'$. 

This means that, when a "shortcut" happens, one transition in $\rho^\eta$ is not matched to one transition in $\rho$, but rather to either a sequence of delayed transitions or a makeshift transition in $\rho$.
%%This is enough to verify $\weight(\rho)\leq \sup_{\sigma'_\Maxsf}\weight
%(\play(\sigma_\Minsf,\sigma'_\Maxsf))$.

%then there exists an increasing chain of plays in $\Pcal'$ with respect to the prefix ordering that converges to this supremum. In this case, we define $\rho^\id$ such that the "shortcut transition" in $\rho$ corresponds to the limit of this chain in $\rho^\id$.

\item Whenever a clock in $\Xflawed$ is reset along $\rho^\eta$, remove it from $\Xflawed$.
\end{itemize}

Then the
 pair $(\rho^\eta,\rho)$ satisfies the following properties:

\begin{itemize}

\item In $\rho^\eta$, between any $d+1$ visits to $\ell$ with positive delay, the total time elapsed between entering $\ell$ on the first visit and entering $\ell$ on the $(d+1)$-st visit is at least ${\eta'}$.

Assume, for the sake of contradiction, that less than ${\eta'}$ time units elapse between these two visits. Then none of the first $d$ visits is followed by a "shortcut transition". Consequently, between every pair of consecutive visits, at least one clock $y$ is reset, and
\begin{itemize}
\item either $y\in\Xflawed$, and therefore removed from $\Xflawed$.
\item or $y$ has valuation $>\eta'$. Being reset, $y$ won't fit this condition again until $\eta'$ time units elapse.
\end{itemize} 
Hence, at least $d$ distinct clocks are reset during these $d+1$ visits. This is impossible, since $\Xcal \setminus {x}$ contains at most $d-1$ clocks.

Thus, since at most $1$ time unit can be spent in an "up-driven" $\intfo01$-\wtg,
the number of positive-delay visits of $\ell$ in $\rho^\eta$ is bounded by
$$
d \cdot \left(\left\lfloor \frac{1}{\eta'}\right\rfloor +1\right).
$$

Therefore $\sigma^\eta_\Minsf$ can indeed be considered as a strategy on both $\Gcal'$ and $\Gcal''$, and satisfies 
$$\sup_{\sigma_\Maxsf}\weight
(\play_{\Gcal''}(\sigma^\eta_\Minsf,\sigma_\Maxsf))= \sup_{\sigma_\Maxsf}\weight
(\play_{\Gcal'}(\sigma^\eta_\Minsf,\sigma_\Maxsf))\,.$$

\item At every step of the construction, the distance between the valuations of $\rho^\eta$ and $\rho$ is at most $2\eta'$. Indeed, the valuations may differ only on the clocks in $\Xflawed$. Before a clock $y$ is added to $\Xflawed$, both copies of $y$ agree on a valuation $\leq \eta'$.
A reset occurs on $\rho$ before at most another \(\eta'\) units of elapsed time, therefore skipping that reset creates discrepancy at most \(2\eta'\).
   Whenever a clock in $\Xflawed$ is reset, the discrepancy on that clock is eliminated. This is why "shortcut" operations are restricted to plays that do not reset any clock in $\Xflawed$, thereby preventing discrepancies from accumulating.

Moreover, every "shortcut" operation incurs a cost no greater than that of the play it replaces. This follows from the fact that $\ell$ is a minimum-weight location among the non-urgent locations of $\Gcal'$, and that every cycle has null transition weight. Consequently, upon leaving $\Gcal'$, the cost of the "shortcut" play $\rho^\eta$ exceeds the cost of  $\rho$ by at most $2\eta' \cdot \slopemax = \eta$.

%Be aware that, for the purposes of this argument, we depart from the convention that infinite plays have cost $+\infty$. Instead, we interpret the cost of $\rho^\id$ as the infinite sum of its transition costs.
%limit of the costs of its finite prefixes.
\item $\rho$ verifies the required property $\weight(\rho)\leq \sup_{\sigma'_\Maxsf}\weight
(\play(\sigma_\Minsf,\sigma'_\Maxsf))$.
%\item 
%There exists a strategy $\sigma_\Maxsf$ such that $\rho^\id=\play_{\Gcal'}(\sigma^\id_\Minsf,\sigma_\Maxsf)$.
\end{itemize}
Therefore,
\begin{align*}
\mathsf{Val}_{\Gcal''} 
&=
  \sup_{\sigma_\Maxsf}\inf_{\sigma_{\Minsf}}  
\mathsf{weight}(\play_{\Gcal''}(\sigma_{\Minsf},\sigma_{\Maxsf}))\\
&\leq \sup_{\sigma_\Maxsf}\inf_{\sigma_{\Minsf}}
 \mathsf{weight}(\play_{\Gcal''}(\sigma^\eta_{\Minsf},\sigma_{\Maxsf}))                                             
                                              \\
&= \sup_{\sigma_\Maxsf}\inf_{\sigma_{\Minsf}}
 \mathsf{weight}(\play_{\Gcal'}(\sigma^\eta_{\Minsf},\sigma_{\Maxsf}))                                             
                                              \\
&= 
\sup_{\sigma_\Maxsf}\inf_{\sigma_{\Minsf}}
\weight(\rho^\eta)
\leq  
\sup_{\sigma_\Maxsf}\inf_{\sigma_{\Minsf}}
\weight(\rho)+\eta\\
&\text{ with $(\rho^\eta,\rho)$ generated for fixed $\sigma_\Minsf,\sigma_\Maxsf$}                                       \\
 &\leq
 \sup_{\sigma_\Maxsf}\inf_{\sigma_{\Minsf}}\sup_{\sigma_{\Maxsf'}}
 \mathsf{weight}(\play_{\Gcal'}(\sigma_\Minsf,\sigma_{\Maxsf'})) + \eta\\
 &=
 \sup_{\sigma_\Maxsf} \mathsf{Val}_{\Gcal'}+ \eta=
\mathsf{Val}_{\Gcal'}+\eta
\end{align*}

%%%%%%%%%%%%%%%%%%%%%%%%%%%%%%%%%%%%%%%%%%%%%%%%%%%%%%%%%%%%%%%%%%%%
%ELL BELONGS TO MAX
%%%%%%%%%%%%%%%%%%%%%%%%%%%%%%%%%%%%%%%%%%%%%%%%%%%%%%%%%%%%%%%%%%%%
\item[If $\ell$ belongs to $\Maxsf$,]
consider $\Gcal'''$ a copy of $\Gcal''$ such that
$\Maxsf$ is forced to pick either delay $0$ or $\geq \frac\eta{2\slopemax}$ in $\ell$~\footnote{Non-integer guards are usually not supported by \wtgs. However, the goal in this section is to transform a \wtg into an acyclic \wtg, whose value is computable even with non-integer guards.}.

Then
% ${\Gcal'}$ offers more freedom to $\Maxsf$ than ${\Gcal''}$, therefore
$\Value{\Gcal'}\geq \Value{\Gcal''}\geq \Value{\Gcal'''}$.
Indeed, strategies for $\Minsf$ can be translated between $\Gcal'$ and $\Gcal''$ in both directions, whereas strategies for $\Maxsf$ in $\Gcal''$ can be translated to $\Gcal'$ but not conversely.
We reason similarly for $\Gcal''$ and $\Gcal'''$.
 Note that the additional limitation imposed on $\Maxsf$ in $\Gcal'''$ compared to $\Gcal'$ is due to the non-integer guards on transitions between copies, rather than to the number of copies themselves. Indeed, within a single time unit, a play can traverse at most $\frac{1}{\eta'} < K_\eta$ transitions between copies.

To prove the other inequality, fix $\sigma_\Minsf$ a $\Minsf$ strategy in $\Gcal'''$, and $\sigma_\Maxsf$ a $\Maxsf$ strategy in $\Gcal'$.
Let us define the following strategies $\sigma^\eta_\Minsf$ in $\Gcal'$ and $\sigma^\eta_\Maxsf$ in $\Gcal'$:
%$$\sup_{\sigma_\Maxsf}
%\mathsf{weight}(\play_{\Gcal'}(c,\sigma_{\Minsf},\sigma_{\Maxsf}))
%-
%\sup_{\sigma_\Maxsf}
%\mathsf{weight}(\play_{\Gcal''}(c',\sigma_{\Minsf^\id},\sigma_{\Maxsf}))
%\leq \eta
%$$
%for any configuration $c$ of $\Gcal'$, where $c'$ is $c$ in the first $\Gcal'$ copy of $\Gcal''$.                       
\begin{itemize}
\item As before, we define $\sigma^\eta_\Minsf$ iteratively by constructing simultaneously a play $\rho^\eta=\play_{\Gcal'}(\sigma^\eta_\Minsf,\sigma_\Maxsf)$, and a play $\rho=\play_{\Gcal'''}(\sigma_\Minsf,\sigma^\eta_\Maxsf)$, both starting in the same configuration. Our construction guarantees that the two plays traverse the same path $\pi$.
However, this construction does not enforce that $\rho^\eta$ and $\rho$ be synchronized; the upper clocks of $\rho^\eta$ can be ahead of those of $\rho$. 
 \item The variable $\Delta$ stores the difference of elapsed time between the two plays. Initially, let $\Delta=0$. 
%  \item Whenever $\Maxsf$ takes a delayed transition along $\rho$ from some location $\ell_\Maxsf\neq \ell$,  with a constraint $x=1$,
% we define the corresponding play $\rho^\id$ so that $\Maxsf$ takes the same transition with a valid delay.
 \item Whenever $\sigma_\Maxsf$ prescribes a delayed transition $(\delta,t)$ along $\rho^\eta$ from some location $\ell_\Maxsf\neq \ell$, 
let $\sigma_\Maxsf^\eta$ prescribes $(\delta,t)$. We assume that $t$ has no guard of the form $(x=1)$(these can be lumped with output transitions). 
Since the upper clocks in $\rho_\eta$ can only be ahead of $\rho$, guards of the form $x\leq 1$ are verified. Thus, because of trimming and relaxed guards, this delayed transition is valid. 
\item Whenever $\sigma_\Maxsf$ prescribes a delayed transition from $\ell$ along $\rho^\eta$, with delay either $0$ or $\geq\eta'$ from $\ell$ in ${\Gcal'}$, let $\sigma^\eta_\Maxsf$ prescribes the same delayed transition.
\item However, when $\sigma_\Maxsf$ prescribes a delayed transition from $\ell$ along $\rho^\eta$ with delay  $0<\delta<\eta'$, this transition cannot be reproduced with the same delay in $\rho$. 

Now, if $\Delta+\delta<\eta'$, we define $\rho$ such that $\sigma^\eta_\Maxsf$ prescribes the corresponding transition with delay $0$ in $\Gcal'''$ (thus staying in the same copy). Since $\Gcal'$ is "relaxed", this transition is valid. Add $\delta$ to $\Delta$.

Otherwise, $\Delta+\delta\geq \eta'$. Then let us define $\rho$ such that $\sigma^\eta_\Maxsf$ prescribes the corresponding transition with delay $\Delta+\delta$ (thus jumping to the next copy). 
This delayed transition is valid since it synchronizes the pair $(\rho^\eta,\rho)$. 
 Reset $\Delta$ to $0$.

\item
Whenever $\sigma_\Minsf$ prescribes a transition of delay $\delta>0$ in $\Gcal'''$, $\sigma^\eta_\Minsf$ prescribes the corresponding transition with delay $\delta'=\max (0,\delta-\Delta)$. Relaxed guards and trimming in $\Gcal'$ ensure that this delayed transition is valid. Let $\max(0,\Delta-\delta)$ be the new value in $\Delta$.

\end{itemize}

The pair $(\rho^\eta,\rho)$ satisfies the following properties:
\begin{itemize}
\item The difference between a pair of clock valuations along $(\rho^\eta,\rho)$  never exceeds $\eta'$: 
Indeed,
consider a pair of valuations $(\nu^\eta,\nu)$ at some step of $(\rho^\eta,\rho)$. First observe that for any clock $y$, $\nu^\eta(y)-\nu(y)\leq \nu^\eta(x)-\nu(x)=\Delta\leq \eta'$ by construction.

 Now let us show that $\nu(y)-\nu(y)^\eta\leq
 {\eta'}-\Delta$. For the sake of contradiction, assume that 
\begin{align*}
\nu(y)-\nu(y)^\eta&>
 {\eta'}-\Delta\\
&=  {\eta'}-(\nu^\eta(x)-\nu(x)),\\
\text{i.e., }
 \nu^\eta(x) -\nu^\eta(y) &>
{\eta'} + \nu(x) -\nu(y).
\end{align*}
Since the gap $\nu^\eta(x)-\nu^\eta(y)$ increases only when $y$ is reset, the above inequality implies that there must have been a previous pair of corresponding configurations $(\nu'^\eta,\nu')$ such that
$\nu'^\eta(y)=\nu'(y)=0$ and
$
\Delta=\nu'^\eta(x)-\nu'(x)>\eta'
$,
contradicting the previous claim.

\item If $(\rho^\eta,\rho)$ ends with a time difference $\Delta$, then 
\begin{align*}
\weight(\rho^\eta)&\leq \weight(\rho) + \Delta\cdot w(\ell) \\
&\leq \weight(\rho) + \eta'\cdot \wmax, 
\end{align*}

since every delay taken by $\rho^\eta$ in $\ell$, a location of least weight, was compensated in $\rho$, except for $\Delta$ time units.

Taking into account the output functions, we obtain
$$\weight(\rho^\eta)\leq  \weight(\rho) + \eta'\cdot \wmax +\eta'\cdot \slopemax \leq \weight(\rho)+\eta.$$
%\item There exists a strategy $\sigma_\Maxsf$ such that $\rho^\id=\play_{\Gcal''}(\sigma^\id_\Minsf,\sigma_\Maxsf)$.

%In the case where $(\rho,\rho^\id)$ ends with a $x=1$ constraint, 
%let $\Delta'$ be the last $\Delta>0$ seen along the path (or $0$ if it does not exists. Then
%\begin{align*}
%\weight(\rho)&\leq \weight(\rho^\id) + \Delta'\cdot w(\ell) +\Delta'\cdot\wmax\\
%&\leq \weight(\rho^\id) + 2\eta'\cdot \wmax\leq \weight(\rho^\id)+\eta  
%\end{align*}

%\item $\rho^\id$ verifies the "ideal strategy invariant", thus $\Value{\Gcal'''}\geq \weight(\rho^\id)$.

\end{itemize}

Therefore, 

\begin{align*}
\mathsf{Val}_{\Gcal'} &=
    \sup_{\sigma_\Maxsf}\inf_{\sigma_{\Minsf}}
                                              \mathsf{weight}(\play_{\Gcal'}(\sigma_{\Minsf},\sigma_{\Maxsf}))\\
&\leq \sup_{\sigma_\Maxsf}\inf_{\sigma_{\Minsf}}
                                              \mathsf{weight}(\play_{\Gcal'}(\sigma^\eta_{\Minsf},\sigma_{\Maxsf}))                                             
                                              \\
%                                              
%&=\sup_{\sigma_\Maxsf}\inf_{\sigma_{\Minsf}}
%                                              \mathsf{weight}(\rho^\eta)
%                                               \leq
%\sup_{\sigma_\Maxsf}\inf_{\sigma_{\Minsf}} \mathsf{weight}(\rho) + \eta\\
%                                              \\
&\leq
 \sup_{\sigma_\Maxsf}\inf_{\sigma_{\Minsf}}
                                              \mathsf{weight}(\play_{\Gcal'''}(\sigma_{\Minsf},\sigma^\eta_{\Maxsf})) + \eta\\
                                              &\leq
 \sup_{\sigma_\Maxsf}\inf_{\sigma_{\Minsf}}
                                              \mathsf{weight}(\play_{\Gcal'''}(\sigma_{\Minsf},\sigma_{\Maxsf})) + \eta
                                             =
\mathsf{Val}_{\Gcal'''}+\eta \qedhere
\end{align*}

Therefore $\Value{\Gcal'}-\eta  \leq \Value{\Gcal''}\leq \Value{\Gcal'}$.

%Note that in $\Gcal'''$, $\Maxsf$ can jump copies at most $
%\left\lfloor \frac{1}{\eta'} \right\rfloor
%$ times in one time unit, so the string of copies is long enough to allow that.

\end{description}

%In both cases, we obtain a decomposition of the"SCC"into a chain of
%$$
%(d+1)\cdot \slopemax \cdot \left\lceil \frac{1}{\eta}\right\rceil
%$$
%copies, where waiting in $\ell$ corresponds to moving to the next copy of the chain.
%
%Moreover, the value of this chain is at most the value of the original "SCC" plus $\eta$. Finally, every "SCC" occurring in the chain contains only $n-1$ non-urgent locations. Hence, the induction hypothesis applies.
\end{proof}
%%%%%%%%%%%%%%%%%%%%%
%INDUCTION
%%%%%%%%%%%%%%%%%%%%
Let us now inject \cref{dealing-with-least-weight-state} into an induction.

\begin{restatable}[SCC unfolding]{lemma}{SCCinduction}
\label{lem-induction}

Let $m$ be the number of locations of $\Gcal'$, and $n$ the number of "non-urgent" locations.

Then, for any $\eta>0$, there exists an acyclic WTG $\Gcal''$  of 
tree depth $\Nind{\eta,n,m}$, with non-integer guards on  some $\Maxsf$ transitions, such
that $\distinf{\Value{\Gcal'}(\ell'_0,\cdot)-\Value{\Gcal''}(\ell''_0,\cdot)}\leq \eta\cdot n$.
\end{restatable}

\begin{proof}
By induction on the number of "non-urgent" locations in $\Gcal'$.

In an all-urgent finite-valued SCC with zero internal transition weights, $\Minsf$ has a strategy reaching an output within at most \(m\) steps.
Therefore $\Gcal'$ can be unfolded into an equivalent acyclic \wtg of tree depth $m$.

If there are $n$ "non-urgent" locations, then according to \cref{dealing-with-least-weight-state}, $\Gcal'$ can be unfolded into a string of $
K_\eta +1
$ copies of itself with $n-1$ "non-urgent" locations in each copy, for a value difference at most $\eta$. Fix $\eta'=\frac\eta{K_\eta+1}$. 
Then by induction hypothesis, any copy in the string can be unfolded into an acyclic \wtg of tree depth $\Nind{\eta',n-1,m}$ with value difference $\eta'\cdot (n-1)$.
Thus $\Gcal'$ can be unfolded into an acyclic \wtg of tree depth $\Nind{\eta,n,m}$ with value difference $\eta\cdot n$.\qedhere

\end{proof}

To conclude the proof of \cref{VIA-convergence-rate}, it remains to observe that the
existence of such an acyclic \wtg{} $\Gcal''$ implies the existence of a
$(\eta\cdot n)$-optimal strategy for $\Minsf$ in $\Gcal'$ whose plays have length at
most $\Nind{\eta,n,m}$. Indeed, every strategy in $\Gcal''$ naturally induces a strategy
in $\Gcal'$ by folding the copies back together. This argument relies crucially on the
construction of \cref{dealing-with-least-weight-state}: when the location of minimum
weight belongs to $\Maxsf$, the derived $\eta$-optimal strategy for $\Minsf$ preserves
the length bound on plays. Consequently, unfolding the SCC does not introduce longer
plays, and the bound $\Nind{\eta,n,m}$ carries over to the original game.

%\todo{we might have a SCC problem here}

\subsection{Proof of \cref{unfolding-N}}

\unfoldingN*

\begin{proof}

Observe that $K_\eta +1 \leq \frac{2d\cdot\slopemax}{\eta} +d+1
\leq \frac{2d\cdot\slopemax+d+1}{\eta} 
$ if $\eta<1$.
%Let
%\[
%B(x)=\frac{C}{x}+D,
%\]
%and define
%\[
%f(x,0)=A,
%\qquad
%f(x,n)=B(x)\,f\!\left(\frac{x}{B(x)},n-1\right).
%\]

Introduce the sequence
\[
x_0=\eta,
\qquad
x_{k+1}=\frac{x_k}{K_{x_k}+1}
\geq \frac{x_k^2}{2d\cdot\slopemax+  (d+1)}.
\]

By repeatedly unfolding the recurrence, we obtain
\[
\Nind{\eta,n,n}
=
\prod_{k=0}^{n-1} \left(K_{x_k}+1\right)\cdot
\Nind{x_n,0,n}
=
\prod_{k=0}^{n-1} \left(K_{x_k}+1\right)\cdot
n.
\]

By definition of \(x_{k+1}\),
$
K_{x_k}+1= \frac{x_k}{x_{k+1}}
$.
Therefore,
\[
\prod_{k=0}^{n-1} (K_{x_k}+1)
=
\prod_{k=0}^{n-1}\frac{x_k}{x_{k+1}}
=
\frac{x_0}{x_n}.
\]

Now, defining
$
y_k=\frac{x_0}{x_k}
$,
we have
$
y_0=1
$,
and
$$
y_{k+1} 
\leq y_k^2\cdot \frac{2d\cdot\slopemax+d +1}{\eta}
\;\text{ thus }\;
y_n\leq \left(\frac{2d\cdot\slopemax+d +1}{\eta}\right)^{2^{n}-1}
$$

and therefore
%\[
%\Nind{\eta,n,n}\leq |L|\cdot \left(\frac{2d\cdot\slopemax+d +1}{\eta}\right)^{2^{n}-1}.
%\]
\[
\Nind{\eta,n,n}\leq \left(\frac{2d\cdot\slopemax+d +1}{\eta}\right)^{2^{n}-1}\cdot n.
\]

Thus
\begin{align*}
N'_\epsilon&= \Nind{\epsilon_n, n,n} \cdot n \\
&\leq 
 \left(\frac{2d\cdot\slopemax+d +1}{\epsilon_n}\right)^{2^{ n}-1}
 \cdot n^2 \\
&= \left(\frac{2d\cdot\slopemax+d +1
}{\epsilon}\right)^{2^{ n}-1}
\cdot n^{2^{n+1}}\,.
\end{align*}\qedhere

\end{proof}

\end{document}